\documentclass[11pt]{article}

\usepackage{graphicx}
\usepackage{xcolor}
\usepackage{hyperref}
\usepackage{subcaption}

\usepackage{amsmath,amssymb,epsf,epic}
\usepackage{tikz,tkz-tab}

\numberwithin{equation}{section}

\newcommand{\un}{{\rm Id}}
\newcommand{\ra}{\rightarrow}
\newcommand{\tr}{\mbox{Tr}}

\newcommand{\bra}{\langle} 

\newcommand{\ket}{\rangle}

\renewcommand{\i}{{\rm i}}

\newcommand{\be}{\begin{equation}}
\newcommand{\ee}{\end{equation}}
\newcommand{\bea}{\begin{eqnarray}}
\newcommand{\eea}{\end{eqnarray}}

\newcommand{\eps}{\varepsilon}

\newcommand{\ffi}{\varphi}

\newcommand{\ode}{{\cal O}}
\newcommand{\e}{{\rm e}}
\newcommand{\grintl}{[\kern-.18em [}
\newcommand{\grintr}{]\kern-.18em ]}

\newcounter{resultcounter}[section]

\newtheorem{thm}[resultcounter]{Theorem}
\newtheorem{lem}[resultcounter]{Lemma}
\newtheorem{prop}[resultcounter]{Proposition}
\newtheorem{cor}[resultcounter]{Corollary}

\newtheorem{rem}[resultcounter]{Remark}

\newtheorem{ex}[resultcounter]{Example}

\def\cD{{\cal D}}  
 \def\cH{{\cal H}} 
 \def\cK{{\cal K}} \def\cL{{\cal L}}

 \def\cT{{\cal T}} 
  \def\cX{{\cal X}}

\newcommand{\R}{{\mathbb R}}
\newcommand{\N}{{\mathbb N}}

\newcommand{\C}{{\mathbb C}}
\newcommand{\Z}{{\mathbb Z}}

\def\proof{\noindent{\bf Proof:}\ \ }
\def\qed{\hfill $\Box$\medskip}

\begin{document}
\title{A Nonlinear Quantum Adiabatic Approximation}
 \author{ Clotilde Fermanian-Kammerer
 \footnote{Universit\'e Paris Est Cr\'eteil, Universit\'e Gustave Eiffel, CNRS, LAMA, UMR CNRS 8050,
61, avenue du G\'en\'eral de Gaulle, 94010 Cr\'eteil Cedex, France, 
}\and Alain Joye\footnote{ Univ. Grenoble Alpes, CNRS, Institut Fourier, F-38000 Grenoble, France} }

\date{ }

\maketitle

\maketitle
\vspace{-1cm}

\thispagestyle{empty}
\setcounter{page}{1}
\setcounter{section}{1}

\setcounter{section}{0}

\vskip 1cm

\noindent{\bf Abstract}:  { This paper is devoted to a generalisation of the quantum adiabatic theorem to a nonlinear setting. We consider a Hamiltonian operator which depends on the time variable and on a finite number of parameters and acts on a separable Hilbert space of which we select a fixed basis. We study an evolution equation in which this  Hamiltonian acts on the unknown vector, while depending on coordinates of the unknown vector in the selected basis,  thus making the equation nonlinear.  We prove existence of solutions to this equation and consider their asymptotics in the adiabatic regime, {\it i.e.} when the Hamiltonian is slowly varying in time. Under natural spectral hypotheses, we prove the existence of normalised time dependent vectors depending on the Hamiltonian only, called {\it instantaneous nonlinear eigenvectors}, and show the existence of solutions which remain close to these vectors, up to a rapidly oscillating phase, in the adiabatic regime. We first investigate the case of bounded operators and then exhibit a set of spectral assumptions under which the result extends to unbounded Hamiltonians. }

\medskip 

\noindent{\bf Keywords}: {Adiabatic approximation, nonlinear adiabatic theorem.} 

\medskip 

\noindent {\bf Acknowledgments}: This work is partially supported by the ANR grant NONSTOPS (ANR-17-CE40-0006-01), by the Von Neumann visiting professorship of the  Technische Universit\"at M\"unchen and the CNRS Project 80$|$Prime {\it AlgDynQua}. It has been partially written in the Mathematics Department of the Technische Universit\"at M\"unchen and during a visit of AJ at the Institut Mittag-Leffler. The authors thank these two institutions for their warm hospitality, with a special thought for the stimulating discussions with Simone Warzel and Herbert Spohn.

\section{Introduction}

We consider a time dependent Hamiltonian on a separable Hilbert space $\cH$ that depends on a finite number of real parameters taken  { in some open neighborhoods $\mathcal T$ and $\mathcal X$ of $[0,1]$:}
\be\label{defH}
\mathcal T \times \mathcal X^p\ni (t,x)\ra H(t,x)\in \cL(\cH),
\ee
where $H(t,x)$ is a   {smooth map, {\it i.e.} ${\mathcal C}^\infty$,} valued in the set of  self-adjoint operators on $\cH$. Let $\{e_j\}_{j\in \N}$ be a fixed orthonormal basis  of $\cH$ and for $f\in\cH$, we denote by $f_j$ its coordinate along  $e_j$, i.e. $f_j=\langle e_j | f\rangle$. 
 {We consider  solutions to the  nonlinear evolution equation 
\be\label{nlad}
i\eps \partial_t v^\eps(t)=H\Bigl(t,  |v_1^\eps(t)|^2, \dots, |v_p^\eps(t)|^2\Bigr)v^\eps(t), \ \ v^\eps(0)= v_0, 
\ee
 for $t\in \mathcal T$ and initial data $v_0\in \cH$ with  $\| v_0\|=1$,
in the  limit where the small parameter $\eps$ tends to zero. }

\medskip 

More precisely, we prove   {under natural spectral hypotheses} that for the systems  we consider, there exist an interval of times ${\mathcal T}_0\subset \mathcal T$ (containing $0$) and a  family of smooth nonlinear eigenvectors, {\it i.e.} two smooth maps $t\mapsto \omega(t)\in \cH$ and $t\mapsto \lambda(t)\in \R$, such that $\|\omega(t)\|\equiv1$, $\bra \omega(t)|\partial_t\omega(t)\ket\equiv 0$, and 
$$H(t, |\omega_1(t)|^2,\dots  ,|\omega_p(t)|^2)\omega(t)=\lambda(t)\omega(t),\;\;   {\forall t\in \mathcal T_0}$$
and we provide conditions under which the deviations  of $v^\eps(t)$ from $\omega(t)\e^{-{i\over\eps} \int_0^t \lambda(s) ds}$ 
are small as $\eps\ra 0$, in the case where the initial data is taken along $\omega(0)$ ($v^\eps(0)=\omega(0)$).  

\medskip

We stress that the evolution equation (\ref{nlad})  depends on the choice of  {the first $p$ vectors of the orthonormal basis $\{e_j\}_{j\in \N}$. It is also important to note that the norm of $v^\eps(t)$ is preserved and, because of the choice of a normalized initial data, $\| v^\eps(t)\|.$ In particular, this implies 
$$|v^\eps_j(t)|^2\in[0,1] \subset \mathcal X,\;\;\forall j\in\{1,\cdots p\}.$$  
}

 {
The limit $\eps\rightarrow 0$ that we consider is known as the {\it adiabatic limit} and consists in analyzing in finite time the evolution of slowly varying Hamiltonian: indeed, with the change of variable $t=\eps s$ and of unknown function $\tilde v^\eps(s)=v^\eps(t)$, equation (\ref{nlad}) is equivalent to
\[ i\partial_s \tilde v^\eps = H\Bigl(\eps s, |\tilde v_1^\eps(s)|^2, \dots, |\tilde v_p^\eps(s)|^2\Bigr)\tilde v^\eps(s), \ \ \tilde v^\eps(0)= v_0\]
where the map $s\mapsto H(\eps s,x)$ is slowly varying in $s\in \mathcal{ T}/\eps $. In the context of linear equations, such an analysis leads to the celebrated adiabatic theorem of quantum mechanics see {\it e.g.}~\cite{K1}. 
Our aim here is to provide a framework where one can prove an approximation of the 
solution to the nonlinear equation (\ref{nlad}) that bears some similarities with the well-known adiabatic results of the litterature.
}

\medskip 

The adiabatic theorem of quantum mechanics has found numerous extensions since its first formulations \cite{BF, K1} for self-adjoint time dependent Hamiltonians with an isolated eigenvalue. It was extended to accommodate isolated parts of spectrum \cite{N, ASY} and it was shown to be exponentially accurate for analytic time dependence \cite{JKP, JP, N2, J}.
 Then, it was extended to deal with gapless situations where the eigenvalue of interest is not isolated in the rest of the spectrum, \cite{AHS, AE, T}. Generalisations to non-self-adjoint generators were provided in \cite{A-SF, J2, AFGG}, leading to extensions to gapless, non self-adjoint generators  provided in \cite{Sch}. Also, formulations of the adiabatic approximation have been shown to hold true for unitary and non unitary discrete time evolutions, \cite{DKS, Ta, HJPR1, HJPR2}, and for extended
many body systems \cite{BDF}. From this perspective, we prove a generalisation of the adiabatic theorem to nonlinear non-autonomous evolution equations in a Hilbert space defined by (\ref{nlad}) and (\ref{defH}). 

\medskip 

Such nonlinear evolution equations occur for example in condensed matter Physics or nonlinear Optics 
 within certain parameter regimes. In particular, the analysis of Landau-Zener tuneling of a Bose-Einstein condensate between Bloch bands in an optical lattice or in double well potentials, as in~\cite{BQ} , \cite{PRL03,Kh10,KhRu05} or the study of optical waveguides known as nonlinear coherent couplers \cite{Je,A}, lead to systems of this form.  Indeed, within a certain regime, the relevant Hamiltonians take the explicit form~(\ref{nlad}) for $p=2$ with an explicit two by two matrix~$H(t,x_1,x_2)$,  see the book \cite{LLFY} for examples and more references. 
  {
 A concrete example is provided by the work \cite{MCWW} in which the dynamics of a Bose-Einstein condensate in a double-well potential are studied, within the mean-field and two-mode approximations. In this regime, one considers the Gross-Pitaevskii equation for the condensate, assuming the two wells of the potential are sufficiently deep and separated so that the condensate wave function can be expressed as a linear combination of the ground states in each of the two wells. Under suitable assumptions on the many-body interaction term, the resulting effective evolution of the coefficients $(v_1, v_2)\in {\mathbb C}^2$  of this expansion that describe the number of particles in each well, takes the form
\[
i\partial_t 
\begin{pmatrix} v_1 \cr v_2 \end{pmatrix}= \begin{pmatrix} \kappa |v_1|^2 & \Omega  \cr \Omega &  \kappa |v_2|^2 \end{pmatrix}\begin{pmatrix} v_1 \cr v_2 \end{pmatrix},
\]
where $\Omega$ and $\kappa$ are parameters which depend on time when the two-well potential depends on time. This yields a Hamiltonian of the type (\ref{defH}). While the quadratic dependence in the components $|v_j|$ of the Hamiltonian above is dictated by the Physics it describes, our results hold for Hamiltonians displaying arbitrary smooth dependence in $|v_j|^2$.
 }

\medskip 

Adiabatic issues have been already addressed in the PDE literature in a nonlinear setting  with different perspectives.
With a scattering point of view, the long time behaviour of nonlinear two by two problems with generators similar to those mentioned above was analysed  by \cite{CFK2}. In a PDE setting,  \cite{CFK1} and~\cite{Hari} study the adiabatic propagation of coherent states for  systems of Schr\"odinger equations with a non linearity and \cite{S} considers the adiabatic regime of the nonlinear Schr\"odinger equation for small data. A common feature of these works is that the effective nonlinearity is weak in the sense that it decays with $\eps$. This is not the case in~\cite{GG}  {which studies} a PDE with a nonlinearity  { of order $\ode(1)$
 as $\eps\rightarrow 0$}, for small initial data, 
 { but of size independent of~$\eps$}. The authors consider therein the time dependent Gross-Pitaevskii equation in a potential which varies slowly in time. Under suitable conditions on the potential, a unique ground state exists for the stationary linear equation parametrized by the time variable,  playing the role of a nonlinear  eigenvector in the sense of the previous paragraphs, and the solution to the Gross-Pitaevskii equation is shown to follow the instantaneous ground state, for large times. 

\medskip

 {Our approach here is closer to the latter reference. Indeed, we aim at providing }
 a general functional framework for nonlinear adiabatic evolution equations (\ref{nlad}) and (\ref{defH}), characterised by non linearities  {of order $\ode(1)$} as $\eps\ra 0$ and admitting solutions of norm  {strictly equal to} one, in contrast to the PDE results mentionned above. We then discuss a set  of reasonable spectral hypotheses on $H(t,x)$ allowing us to  provide an approximation of the solutions to (\ref{nlad}) as $\eps\ra 0$, for times $t$ of order one.
Our main result is first proven for bounded Hamiltonians, and then extended  to unbounded $H(t,x)$, under suitable spectral assumptions. In particular, the latter case applies to a certain type of nonlinear Schr\"odinger equation on $L^2(\R)$ that we discuss.

\medskip

Note that the matrix cases considered in \cite{CFK2} or \cite{LLFY} and in the references therein, appear as special cases of those that we consider, whereas our hypotheses excludes the PDE setup considered in \cite{CFK1,Hari, S,GG}. This is due to the fact that the nonlinearity in~(\ref{nlad}) depends on the norm of the projections of the wave function on some subset of the basis vectors of the Hilbert space, and not of the modulus of the wave function itself as in the Gross-Pita{i}evski equation or in Hartree equation.  In this sense, the nonlinearity that we consider is weaker.

\subsection{Setup and main result}

To ease notations, we will write from now on 
\be\label{allev}
H\Bigl(t,  |v_1|^2, \dots, |v_p|^2\Bigr)= :H\left(t, [v]\right),
\ee
for any vector $v\in\cH$, where $H$ depends on $p<\infty$  components of $v$ only. The form of the nonlinearity we choose, depending on the modulus of (certain components of) the solution, is reminiscent of that of the nonlinear Schr\"odinger equation. It entails in particular the fact that $H$ actually depends on $\{v_1, \bar v_1, v_2, \bar v_2, \dots v_p, \bar v_p\}$. 
This motivates the introduction of the anti-unitary complex conjugation $C$ on $\cH$ defined by
\be\label{compconj}
\forall \   v=\sum_j v_j e_j \in \cH, \ \ Cv=\sum_j \bar {v_j} e_j
\ee
to be used later on.  {Note that $C$ depends on the basis $\{e_j\}_{j\in\N}$ that is considered invariant under $C$.} For any $A\in\cL(\cH)$, we define the operator $\bar A=C A C\in \cL(\cH)$ and will call operators such that $\bar A=A$, real operators. 
We will work under the following general hypotheses.

\begin{enumerate}
\item[{\bf H$_0$}] The map 
$\cT \times \cX^p\ni (t,x)\mapsto H(t,x)\in \cL(\cH)$ is $C^\infty$  {in the operator norm topology}, where $\cT$ and~$\cX$ are open neighbourhoods of $[0,1]$. For all $(t,x)\in \cT\times \cX^p$, $H(t,x)=H^*(t,x)$.

\item[{\bf H$_1$}]  There exists $\delta >0$ such that $\|\partial_{x_j} H(x,t)\|\leq \delta $, for all $(t, x)\in  \cT \times \cX^p$ and $j\in \{1,\dots, p\}$.

\item[{\bf H$_2$}] For all  $(t,x)\in \cT \times \cX^p$, the spectrum $\sigma (H(x,t))$  { consists in $N$ distinct eigenvalues $\{\lambda_j(t,x)\}_{j=1}^N$, possibly degenerate,  that are separated from one another by a gap bounded below by $g>0$, uniformly in $(t,x)$. }

\item[{\bf H$_3$}] There exists $1\leq j_0\leq N$ such that $\lambda_{j_0}(x,t)$ is simple. 
\end{enumerate}

Consequently, the corresponding spectral decomposition of $H(t,x)$ reads
\be\label{specdec}
H(t,x)=\sum_{j=1}^N \lambda_j(t,x) P_j(t,x),
\ee
where the orthogonal spectral projectors $P_j(t,x)$ have  {constant rank which may be infinite}, while $\dim (P_{j_0}(t,x))\equiv 1$.  We shall make use of the following facts:  {the maps $(t,x)\mapsto P_j(t,x)$ are ${\mathcal C}^\infty$ and so are $(t,x)\mapsto \lambda_j(t,x)$.  }
Moreover, for $j=j_0$, there exists a  global smooth map $\cT\times \cX^p\ni (t,x)\mapsto \ffi_{j_0}(t,x)\in \cH$ such that  {$\|\ffi_{j_0}(t,x)\|\equiv 1$ and }
$$\forall (t,x)\in\cT \times \cX^d,\;\; H(t,x)\varphi_{j_0}(t,x) = \lambda_{j_0}(t,x) \varphi_{j_0}(t,x).$$
These facts are  briefly discussed in Section~\ref{sec:eigenvector} below.

\medskip

The form of the nonlinearity immediately implies a gauge invariance, 
which will turn out to be crucial later on.
 {Due to~\eqref{allev}}, we have for any $\theta\in \R$, any $v\in \cH$,
\be\label{gauge}
H(t,[\e^{i\theta}v])=H(t,[v]).
\ee
If {\bf H$_2$} and {\bf H$_3$}  hold as well, this implies
\be\label{gauge2}
P_j(t,[\e^{i\theta}v])=P_j(t,[v]), \ \ \lambda_j(t,[\e^{i\theta}v])=\lambda_j(t,[v]), \ \ \ffi_{j_0}(t,[\e^{i\theta}v])=\ffi_{j_0}(t,[v]).
\ee

We first note that  { the self-adjointness of $H(t,x)$  ensures that $\| v^\eps(t)\|=\|v_0\|=1$, whence } the existence of global solutions to (\ref{nlad}) via Cauchy-Lipschitz Theorem. Moreover, gauge invariance (\ref{gauge}) implies symmetries that we exploit below. These elementary properties are stated in the next Lemma with the convention (\ref{allev}). 

\begin{lem}\label{lem:elem}
Under assumption {\bf H$_0$},  {for any $v_0\in\mathcal H$},  the equation 
\be\label{nlad2}
i\eps \partial_t v^\eps(t)=H(t,  [v^\eps(t)])v^\eps(t), \ \ v^\eps(0)=v_0\in \cH, \ t\in \cT,
\ee
admits a unique global solution with   {$\| v^\eps(t)\|=\|v_0\|$}.\\
Besides, given a $C^0$ map $\cT\times \cX^p\ni (t,x)\ra \chi(t,x)\in \R$, and $v^\eps(t)$ a solution to (\ref{nlad2}), the solution to
$$
i\eps \partial_t s^\eps(t)=\Bigl(H\left(t,  [s^\eps(t)]\right)+\chi\left(t, [s^\eps(t)]\right)\un \Bigr)s^\eps(t) , \ \ s^\eps(0)=v_0\in \cH, \ t\in \cT
$$
reads
$$
s^\eps(t)=\e^{-i\int_0^t \chi(u, [v^\eps(u)])du/\eps}v^\eps(t), \ \ \forall \ t\in \cT.
$$
\end{lem}

Our analysis focuses on solutions to (\ref{nlad}) that are tightly related to the simple eigenvalue $\lambda_{j_0}(t,x)$ and associated eigenvector 
$\ffi_{j_0}(t,x)$. Therefore, to simplify the notation, we drop the index $j_0$ for these spectral data from now on. 
We start by introducing a vector $\omega(t)\in \cH$ 
defined in a neighbourhood of $0\in\cT_0$ 
by 
$$
H(t, [\omega(t)])\omega(t)=\lambda(t,[\omega(t)])\omega(t), \ \ \forall \ t\in \cT_0.
$$
As discussed in Section \ref{sec:eigenvector}, this nonlinear equation   {(that does not involve any derivative of $\omega(t)$)} turns out to always have a local nontrivial solution when $\lambda(t,x)$ is a simple eigenvalue of $H(t,x)$.

\begin{prop} \label{prop:eigenvector}
Assume  {\bf H$_0$}, {\bf H$_1$},  {\bf H$_2$}  and {\bf H$_3$}. 
Then, for any $t_0\in \cT$, there exists a neighbourhood $\cT_0\subset \cT$ of $t_0$ such that for all $t\in\cT_0$, there exists a solution $\omega(t)\in\cH$ of norm one to the equation \be\label{defomega}
P(t, [\omega(t)])\omega(t)=\omega(t).
\ee
 { Moreover, there exists $\delta_0>0$ such that  for $\delta\in (0,\delta_0)$, the map $\cT_0 \ni t\mapsto \omega(t)$ is ${\mathcal C}^\infty$ and can be chosen to satisfy 
 \begin{equation}\label{choice_omega}
 \bra\omega(t) | \partial_t \omega(t)\ket\equiv 0,
 \end{equation}
 which makes it unique up to a  constant phase. }
\end{prop}

 { 
In the sequel, we shall always make the choice~\eqref{choice_omega} and we will call such a vector an {\it instantaneous nonlinear eigenvector}.}
We can now give our main statements which establish nonlinear adiabatic theorems in the considered framework.  {We first consider the case ${\mathcal H}=\C^N$.}

\begin{thm}\label{thm:main0} 
Assume  {${\mathcal H}=\C^N$}, {\bf H$_0$}, {\bf H$_1$} with $\delta$  small enough, and suppose that {\bf H$_2$} holds with all eigenvalues being simple.
Moreover, assume that $H(t,x)$ is real, that is $\overline{H}(t,x)= H(t,x)$, 
and generic in the sense that 
$\sigma(H(t,x)-\lambda(t,x))\cap \sigma(-H(t,x)+\lambda(t,x))=\{0\}$.
Let $\omega(t)$ be  { the instantaneous nonlinear eigenvector} defined in Proposition~\ref{prop:eigenvector}
 in a neighbourhood $\cT_0$ of $t_0=0$.
Then the solution $v^\eps(t)$ to~(\ref{nlad}) with $v^\eps(0)=\omega(0)$ satisfies for all $t\in\cT_0$
$$
v^\eps(t)=\e^{-{i\over\eps} \int_0^t \lambda(s,[\omega(s)]) ds}\omega(t)+ O_t(\eps).
$$
\end{thm}

\begin{rem}\label{rem:main0}
i)  { Note that the condition on the smallness of $\delta$ is independent on $\eps$.}\\
ii) The genericity condition always holds if $\lambda(t,x)$ is the ground state or the largest eigenvalue of $H(t,x)$.
\end{rem}

 {
After a reduction to the case where $\lambda(t,x)=0$, the proof of this theorem relies on the analysis of the system satisfied by the element  $(\Delta(t),\overline{\Delta}(t))$ of $\mathcal H \times  \mathcal H$  where 
$$ \Delta(t)=v^\eps(t) -  \omega(t). $$
Setting  $\dot{\phantom{v}}=\partial_t $, a linearization process around $\omega(t)$ shows that the evolution of $\Delta(t)$ is driven by an evolution equation of the form 
\begin{equation}\label{eq:system}
i\eps 
\begin{pmatrix} \dot\Delta \\ \dot {\overline \Delta}\end{pmatrix}
= F(t) \begin{pmatrix}   \Delta \\ {\overline \Delta}\end{pmatrix}
-i\eps \begin{pmatrix} \dot \omega \\ \dot {\overline \omega}\end{pmatrix}
 +\begin{pmatrix}  r^\eps \\ -\overline r^\eps \end{pmatrix} ,\;\;
 \Delta(0)=0,
 \end{equation}
 with $ r^\eps(t) =O(\| \Delta(t)\|^2)$ and
 $F(t)= F_0(t)+G(t)$ for some finite rank non self-adjoint operator $G(t)$ and  
 \begin{equation}
 \label{defF0}
 F_0(t)= \begin{pmatrix}  H(t,[\omega(t)]) & 0 \\ 0 & - \overline H(t,[\omega(t)]) \end{pmatrix}.
 \end{equation}
 The smallness of $\Delta(t)$ is then proved thanks to a careful analysis of this equation in which the spectrum of $F(t)$ plays a crucial role. 
 The conditions on the spectrum of $H(t,x)-\lambda(t,x)$ that are assumed in Theorem~\eqref{thm:main0} allow to ensure that the operator $F(t)$ is semisimple with 
real eigenvalues of constant multiplicity for all $t\in \cT_0$, which is enough to develop an approach {\it \`a la}  Kato
 and prove
that there exist positive constants $c_0, c_1$ such that the norm of the remainder satisfies
\be\label{balancebound0}
\|\Delta(t)\|\leq \min({c_0 t, c_1\eps}), \ \forall t\in \cT_0.
\ee
These arguments are developed in Section~\ref{sec:proofmain} below} and show that  
 the previous theorem is a special case of the following one, which holds in infinite dimension and bounded operators $H(t,x)$. 

\begin{thm}\label{thm:main}  
Assume {\bf H$_0$}, {\bf H$_1$} with $\delta$ small enough, {\bf H$_2$} and {\bf H$_3$}.
Moreover, suppose that $H(t,x)$ is real, that is $\overline{H}(t,x)= H(t,x)$.
Let $\omega(t)$ be the instantaneous nonlinear eigenvector defined by Proposition~\ref{prop:eigenvector} in a neighbourhood $\cT_0$ of $t_0=0$.
Provided the operator $F(t)$ defined by (\ref{defF}) below is semisimple with 
real eigenvalues of constant multiplicity for all $t\in \cT_0$, 
the solution $v^\eps(t)$ to~(\ref{nlad}) with $v^\eps(0)=\omega(0)$ satisfies for all $t\in\cT_0$
$$
v^\eps(t)=\e^{-i\int_0^t \lambda(s,[\omega(s)]) ds/\eps}\omega(t)+ O_t(\eps).
$$
\end{thm}

As already mentioned, the assumptions of Theorem~\ref{thm:main0} guarantee the adequate spectral behavior of the operator $F(t)$ defined by (\ref{defF})  to get the conclusion of Theorem~\ref{thm:main}. In other words,   assuming in {\bf H$_2$} that all eigenvalues of the real operator $H(t,x)$ are of multiplicity one  is enough to obtain the assumption on the spectral decomposition of $F(t)$. In Section~\ref{sec:proofmain} we describe another set of assumptions which are sufficient to satisfy the hypothesis of Theorem~\ref{thm:main}  {in infinite dimension}  in the case $p=1$, see Lemma~\ref{lem:specft2}.

\subsection{Extension of the result to unbounded operators}

We now extend our results to the case where  the operator $H(t,x)$ on the separable Hilbert space $\cH$ is unbounded and takes the form $H(t,x)=H_0+W(t,x)$, with $W(t,x)\in\cL(\cH)$. We make the following regularity hypothesis:
\begin{enumerate}
\item[{\bf R}$_{0}$]  The self-adjoint operator $H_0$ is defined on a dense domain $\cD\subset \cH$, and the family of bounded operator $W(t,x)$ is self-adjoint for all  $(t,x)\in \cT \times \cX^p$. Moreover, $H_0$, and $W(t,x)$ are real operators.
\item[{\bf R}$_{1}$]  The map $\cT \times \cX^p\ni (t,x)\mapsto W(t,x)\in\cL(\cH)$ is strongly $C^\infty$.
\item[{\bf R$_2$}]  There exist $\delta >0$ such that 
$\| W(t,x)\|\leq \delta$,  $\|\partial_{x_j} W(t,x)\|\leq \delta $, for all $(t, x)\in  \cT \times \cX^p$ and $j\in \{1,\dots, p\}$.
\end{enumerate}
We also assume the  spectral hypothesis 
\begin{enumerate}
\item[{\bf S}$_1$] The spectrum of $H_0$ consists in an infinite increasing 
sequence of 
simple eigenvalues $\lambda_j\geq 0$, $j\in\N$, and there exists $c_0>0$ and $\alpha>1/2$ 
such that the gaps satisfy 
$$\forall j\in\N,\;\; \lambda_{j+1}-\lambda_j  \geq   c_0 \, j^\alpha.$$
\end{enumerate}

The operator $W(t,x)$ being bounded, if $\delta$ is small enough, perturbation theory implies that
for all  $(t,x)\in \cT \times \cX^p$, the self-adjoint operator $H(t,x)=H_0+W(t,x)$ defined on $\cD$ has spectrum $\sigma(H(t,x))=\{\lambda_j(t,x)\}_{j\in \N}$ consisting in simple eigenvalues $\lambda_j(t,x)$ only, and there exists $c_1>0$  
such that the gaps satisfy for $\alpha>1/2$
$$\forall (t,x)\in{\mathcal T}\times \cX^p,\;\; \forall j\in\N,\;\; \lambda_{j+1}(t,x)-\lambda_j(t,x)  \geq   c_1\, j^\alpha.$$
We pick some $j_0\in \N$ and assume the generic property:
\begin{enumerate}
\item[{\bf S}$_2$] 
For all $(t,x)\in {\mathcal T}\times \cX^p$, 
 $\displaystyle{\{\lambda_j(t,x)-\lambda_{j_0} (t,x),\,j\in\N\}\cap \{ -\lambda_j (t,x) +\lambda_{j_0}(t,x),\,j\in\N \}=\{0\}.}$
\end{enumerate}
Note that, since $H_0$ is bounded from below, this assumption concerns only a finite number of eigenvalues. Besides,  this property can be inherited from a similar assumption  on the eigenvalue $\lambda_{j_0}$ of $H_0$. 

\medskip

We consider 
for all $(t_0,x_0)\in{\mathcal T} \times {\mathcal X}^p$, 
the ${\mathcal C}^\infty$ map $(t,x)\mapsto \varphi(t,x)$ from $\cT\times\cX ^p $ to $\cD\subset \cH$ such that 
$$H(t,x)\varphi(t,x)= \lambda_{j_0}(t,x) \varphi(t,x).$$
We drop the index $j_0$ as before. Provided with these properties, we can develop the same analysis as in the situation addressed above, namely, the existence of a nonlinear eigenvector and an adiabatic approximation for the nonlinear evolution equation associated with $H(t,x)$:
We consider $p$ orthonormal vectors $\{e_1, \dots, e_p\}$  {that we take in $\cD$ for convenience},
and set 
$$\forall \psi\in \cH,\;\; [\psi]= ( |\bra e_1| \psi\ket|^2 ,\cdots, |\bra e_p| \psi\ket|^2 ).$$
Proposition \ref{prop:eigenvector} 
ensures that  for any $t_0\in \cT$, there exists a neighbourhood $\cT_0\subset \cT$ of $t_0$ such that for all $t\in\cT_0$, a solution $\omega(t)\in\cD$ of norm one to the algebraic equation~(\ref{defomega})
exists, see Remark~\ref{rem:unbounded}. 
 Moreover $\cT_0 \ni t\mapsto \omega(t)$ is ${\mathcal C}^\infty$ and can be chosen to satisfy $\bra\omega(t) | \dot \omega(t)\ket\equiv 0$.
Taking initial data $\omega(0)$ in~(\ref{nlad}) gives the 
  equation in which  we are interested, namely
\begin{equation}\label{eq:schroabs}
i\eps\partial_t \psi ^\eps (t) =  \Bigl(H_0+W\left(t,[\psi^\eps(t)]\right)\Bigr) \psi^\eps(t),\;\;\psi^\eps(0)= \omega(0),
\end{equation}
in the weak sense on $\cD$.   By solution in the weak sense on $\cD$ we mean the following, see \cite{RS}, vol. II, p. 284 for the linear case: 
For any $\chi\in \cD$,
\be\label{weakschroabs}
\i\eps\partial_t \bra \chi | \psi ^\eps (t)\ket=\bra (H_0 +W(t,[\psi^\eps(t)])) \chi | \psi ^\eps (t)\ket, \ \  \psi ^\eps (0)= \omega (0).
\ee

\begin{thm}\label{thm:unbounded}
\begin{enumerate}
\item Assume {\bf R$_0$} and  {\bf R$_1$}, then equation~(\ref{eq:schroabs}) admits a unique global solution in the weak sense of norm one.
\item Assume moreover  {\bf R$_2$} with $\delta$ small enough, {\bf S}$_1$ and {\bf S}$_2$ and let $\omega(t)$ be a  ${\mathcal C}^\infty$ solution to (\ref{defomega}) in a neighbourhood $\cT_0$ of $t_0=0$.
Then the solution $v^\eps(t)$ to~(\ref{nlad}) with $v^\eps(0)=\omega(0)$ satisfies for all $t\in\cT_0$
$$
\psi^\eps(t)= {\rm e}^{-{i\over \eps} \int_0^t \lambda(t,[\omega(t)]) dt} \omega(t) + O_t(\eps).
$$
\end{enumerate}
\end{thm}

 {
The proof of Theorem~\ref{thm:unbounded} contains two things: the existence of global solutions in the weak sense, and an adiabatic approximation. The proof of the latter follows the same strategy as the one developed in Theorems~\ref{thm:main0} and~\ref{thm:main}. However, additional difficulties  come from the fact that  the spectrum of $F_0(t)$ (as defined in~\eqref{defF0}) consists now in an infinite sequence of eigenvalues, while one has to work in the weak topology and to be careful with domain issues when constructing Kato's operators. These points are documented in Section~\ref{sec:gen} below.} 

 {Before closing this section and discussing properties of these adiabatic solutions, we give a concrete example satisfying the assumptions of Theorem~\ref{thm:unbounded}.}

\begin{ex}
Consider ${\mathcal H}= L^2(\R_y)$ and the operator
$$H_0=-{1\over 2} \Delta_y +V_0(y)$$
with domain ${\mathcal D}\subset L^2(\R_y)$, 
 {where $V_0$ is a polynomial in $|y|$ with highest degree $\beta >6$.
Then, as revealed by the Bohr-Sommerfeld formula~\cite{V}, $H_0$ satisfies the assumptions~{\bf R}$_0$  and~{\bf S}$_1$ above.}
Consider  $x$-dependent self-adjoint perturbations of this operator ($x\in{\mathcal X}^p$)
$$H(t,x)= -{1\over 2} \Delta_y +V_0(y) + W(t,y,x)$$ 
where $W$ is such that 
the map $(t, y, x)\mapsto W(t,y,x)$ is  a bounded function from ${\mathcal C}^\infty(\cT\times \R_y \times \cX^p, \R)$ 
 and there exists $\delta >0$ such that $$\forall(t, y,x)\in  \cT \times \R_y\times \cX^p,\;\;\forall j\in \{1,\dots, p\},\;\;|W(t,y,x)|+  |\partial_{x_j} W(t,y,x)|\leq \delta.$$
 Then, $H(t,x)$ satisfies assumptions {\bf R}$_1$ and {\bf R}$_2$ above. 

\medskip 

 {
To be quite concrete, we can take as orthonormal basis $\{e_k\}_{k\in\N^*}$ of $L^2(\R)$  the set of eigenstates of the harmonic oscillator, $V_0(y)=y^8$, $p=2$ and $W(t,y,x)=-(y-a(t)x_1)^2b(t)e^{x_2}$, where $a, b\in C^\infty(\R,\R)$ with $b(t)>0$. Depending on the sign of $a(t)$, for positive values of $x_1, x_2$, the potential displays one or two wells, of various depths. The nonlinearity manifests itself by emphasising these features, depending on the amplitude of the coefficients of the solution on the first two basis vectors $e_1, e_2$, given by a Gaussian times a Hermite polynomial in $y$.
}
\end{ex}

\medskip

 {The example above is admittedly not motivated by applications to Physics, but is intended to demonstrate the adaptability of our functional framework to the unbounded setup.}

\subsection{ Energy content of the solutions}

We close this introduction by discussing briefly an important feature of the solutions provided by Theorems \ref{thm:main0}
 and \ref{thm:main}. A physically relevant quantity for the nonlinear equation (\ref{nlad}) we consider is the instantaneous energy content of a solution $v^\eps(t)$, defined for all $t\in \cT_0$ by
$$
 E_{v^\eps}(t)=\bra v^\eps(t) | H(t,[v^\eps(t)])v^\eps(t)\ket.
$$
For bounded operators $\cT\times \cX^p\ni (t,x)\mapsto H(t,x)\in \cL(\cH)$, and $\eps-$independent initial conditions $v^\eps(0)=v(0)$ the energy content satisfies the uniform bound 
$$|E_{v^\eps}(t)| \leq \sup_{(t,x)\cT\times \cX^p}\|H(t,x)\|\|v(0)\|^2.$$ 
For a solution of the form $v^\eps(t)=\e^{-i\int_0^t \lambda(s,[\omega(s)]) ds/\eps}\omega(t)+ O_t(\eps)$, 
the energy content simply coincides, to leading order, with the energy content of the corresponding instantaneous nonlinear eigenvalue 
$$
E_{v^\eps}(t)=\bra \omega(t) | H(t,[\omega(t)]) \omega(t) \ket + O_t(\eps)=\lambda(t,[\omega(t)])+ O_t(\eps).
$$
In general, the behaviour in time of the energy content  of a solution does not necessarily admit such a regular behaviour in the limit $\eps \ra 0$, 
 {which makes this property a  specific feature of the adiabatic solutions.}

\medskip 

Let us illustrate this point on the following simple example. Let $\R\ni t\mapsto \gamma(t)\geq \gamma_0>0$ and consider
$$
H(t,x)=\begin{pmatrix} 0 & \gamma(t)x \\ \gamma(t)x & 0
\end{pmatrix}
$$
on $\cH=\mathbb C^2$. The evolution equation (\ref{nlad}) reads
\be\label{exple}
i\eps\partial_t \begin{pmatrix} v_1 \cr v_2 \end{pmatrix}=H(t,|v_1|^2)\begin{pmatrix} v_1 \cr v_2 \end{pmatrix}=\gamma(t)|v_1|^2\begin{pmatrix} v_2 \cr v_1 \end{pmatrix},  
\ee
with initial conditions $\begin{pmatrix} v_1(0) \cr v_2(0) \end{pmatrix}$, and the energy content of the solutions reads
$$
E_{v}(t)=\gamma(t)|v_1|^2 2\Re (v_1\overline{v_2})(t).
$$
 The corresponding real normalised nonlinear eigenvectors $\omega_\pm(t)$ are time-independent,
$$
\omega_\pm(t)=\frac{1}{\sqrt{2}} \begin{pmatrix} 1 \cr \pm 1 \end{pmatrix}, 
$$
and associated to the eigenvalues $\lambda_\pm(t,[\omega_\pm])=\pm \gamma(t)/2$. Hence, the approximate solutions provided by  Theorem~\ref{thm:main0}
read
\be\label{appex}
v_\pm(t)=e^{\mp{ i\over 2\eps} \int_0^t\gamma(u)du}\omega_\pm(t),
\ee
which turn out to be exact solutions for all $t\in \R$, since $\omega_\pm$ are time-independent. Their energy contents are thus given by
 {
$$
E_{v_\pm}(t)= E_{\omega_\pm}(t)= \pm\gamma(t)/2, 
$$}
 which is $\eps$-independent.
However, for general solutions $v^\eps(t)$ the situation is different, as stated in the next Lemma which is proved in Appendix~B.

\begin{lem}\label{lem:energycontent}
Let $v^\eps(t)$ be a solution of equation~(\ref{exple}) with real-valued initial data such that
$v_1(0)>0$, $v_2(0)\neq 0$. Then the energy content reads
$$E_{v^\eps}(t)=2 \,\gamma(t)v_1(0)^3v_2(0)\left[ \cos^2\left({\aleph(t)\over \eps}\right) +\Big(\frac{v_1(0)}{v_2(0)}\Big)^2\sin^2\left({\aleph(t)\over \eps}\right)\right]^{-1}$$
with $\displaystyle{ \aleph (t)=-v_1(0)v_2(0)\int_0^t\gamma(u)du.}$ 
Hence, $E_{v^\eps}(t)/\gamma(t)$ is actually a function of $\int_0^t\gamma(u)du$, which oscillates between the extremal values $2v_1(0)^3v_2(0)$ and $2v_1(0)v_2(0)^3$ with a period of order $\eps$, unless $v_1(0)/v_2(0)=\pm 1$ in which case it is a constant.
\end{lem}

 {
By contrast, the linear quantum adiabatic theorem implies that the energy content of {\it any} solution is given by an $\eps$-independent weighted sum of instantaneous eigenvalues of the Hamiltonian, to leading order.  More precisely, assume $t\mapsto H(t)$ is {\it independent of $x$} and satisfies the hypotheses of Theorem \ref{thm:main0}. Let $\{\ffi_{j}(t)\}_{j=1}^N$ be an orthonormal basis of instantaneous eigenvectors of $H(t)$ with phases normalised by $\bra \ffi_j(t) | \dot\ffi_j(t)\ket\equiv 0$.  Then, the energy content of $v^\eps(t)$, solution to (\ref{nlad}), which is {\it  linear} in this case, with arbitrary initial condition $v_0$ reads for any $t\in [0,1]$, 
\[E_{v^\eps}(t)=\sum_{j=1}^N |\alpha_j|^2 \lambda_j(t)+\ode(\eps), \ \ \mbox{where} \ \  v_0=\sum_{j=1}^N\alpha_j\ffi_j(0).\]
Indeed, the linear quantum adiabatic theorem \cite{K1} implies  $v^\eps(t)=\sum_{j=1}^N\alpha_je^{-i\int_0^t\lambda_j(s)ds/\eps}\ffi_j(t)+\ode(\eps)$, uniformly in $t\in [0,1]$, hence a direct computation of the energy content yields the above expression, thanks to $\bra\ffi_j(t)|\ffi_k(t)\ket=\delta_{j,k}$.
}

 \subsection{Organisation of the paper} 
 
  {
 We begin by proving the existence of the instantaneous nonlinear eigenvectors in Section~\ref{sec:eigenvector}. We also  discuss the limitation that may occur to their existence. This crucial part of our result is independent of the other sections and can be skipped at first reading.  Then we focus in Sections~\ref{sec:proofmain} and~\ref{sec:gen} on the proofs of the nonlinear adiabatic Theorems to which this article is devoted to. Sections~\ref{sec:proofmain} deals with the case of bounded Hamiltonians, with the proofs of  Theorems~\ref{thm:main0} and~\ref{thm:main}, while we explain in Section~\ref{sec:gen} how to  adapt the arguments to the unbounded setting of Theorem~\ref{thm:unbounded}. Finally, two Appendices are devoted to the discussion of examples: the first one shows a situation where the spectrum of the operator $F(t)$ is not necessarily real-valued  even though $H(t)$ is real, as emphasized in Remark~\ref{rem:appendixA} below; the second one focus on the analysis of the Example~\ref{lem:energycontent}. 
 }


\section{ { Instantaneous nonlinear eigenvectors} }\label{sec:eigenvector}

We focus in this section on the existence of the generalized nonlinear eigenvector $\omega(t)$ defined in Proposition~\ref{prop:eigenvector},  { and that we call instantaneous nonlinear eigenvectors}. We first recall well-known facts in the linear setting, mainly to introduce notations. Then, we explain why a similar result remains true locally in the nonlinear regime we consider and why the obtained eigenvectors may not  exist globally. 

\subsection{ Existence of smooth eigenvectors in the linear adiabatic setting} 
The question of local (and global) existence of  ${\mathcal C}^\infty$ eigenvectors is simple in the linear context. 
 Indeed, with the notations of  Assumption~{\bf H}$_2$ and 
 using Riesz formula on  $C_j(g/2)$, a circle of radius~$g/2$ and center $\lambda_j(t,x)$, 
$$
P_j(t,x)=-\frac{1}{2\pi i}\int_{C_j(g/2)}(H(t,x)-z)^{-1}dz,
$$
one gets that the projectors $P_j(t,x)$'s are  ${\mathcal C}^\infty$ as $H(t,x)$ is. Moreover, 
$\|\partial_{x_j}P_j(t,x)\|\leq 2 \delta/g$.
The finitely degenerate eigenvalues $\lambda_j(t,x)=\tr (P_j(t,x)H(t,x)) $ are thus ${\mathcal C}^\infty$, 
and the same is true if ${\rm Rank}\, P_j(t,x) =\infty$.

\medskip

Considering  $j=j_0$, 
for any $(t_0, x_0)\in {\cT}\times {\cX}^p$,  there exists an open neighbourhood of $(t_0, x_0)$ in which a ${\mathcal C}^\infty$ normalised eigenvector $\ffi_{j_0}(t,x)\in \cH$ exists such that
$$
P_{j_0}(t,x)=|\ffi_{j_0}(t,x)\ket\bra \ffi_{j_0}(t,x)|.
$$
 {Here $|\varphi\rangle\langle \psi | : \mathcal H \rightarrow \mathcal H $ maps $\eta$ to $\varphi \langle \psi | \eta\rangle$.}
More specifically, given $\ffi_{j_0}(t_0,x_0)$ an eigenvector of $H(t_0, x_0)$, the vector 
\be\label{defffiloc}
\ffi_{j_0}(t,x):=\frac{P_{j_0}(t,x)\ffi_{j_0}(t_0,x_0)}{\bra \ffi_{j_0}(t_0,x_0)| P_{j_0}(t,x) \ffi_{j_0}(t_0,x_0)\ket}
\ee
satisfies these conditions for all $(x,t)$ such that $P_{j_0}(t,x) \ffi_{j_0}(t_0,x_0)\neq 0$.

\medskip 

Actually, there exists  {an extension of this local map to} a global ${\mathcal C}^\infty$ map $\cT\times \cX^p\ni (t,x)\mapsto \ffi(t,x)$, which can be viewed as follows. Using the shorthand $p=(t,x)$, 
set $E=\cup_{p\in \cT\times \cX^p}(p,\ffi(p))$, and $\pi :  E\ni (p,\ffi(p)) \mapsto p \in \cT\times \cX^p$, so that $\pi: E\ra \cT\times \cX^p$ defines a rank one vector bundle over the base $ \cT\times \cX^p$. The base being contractible, it is known that the vector bundle is trivial, which is equivalent to the existence of a global~${\mathcal C}^\infty$ frame on the fibres of $E$, see {\it e.g.} \cite{LeP, Sp}. An alternative approach is by explicit construction, making use of the parallel transport operator defined by~(\ref{paraltrans}) below. Passing to spherical coordinates $(t, x)\mapsto (r, \theta)\in \R^+\times {\mathbb S}^{p}$ and integrating the parallel transport operator along $r$, keeping~$\theta$ as parameters, we get a ${\mathcal C}^\infty$ unit eigenvector  for each $(t,x)\in \cT\times \cX^p$, by the smoothness of the eigenprojector. This property holds for $\dim \cH =\infty$.

\subsection{Existence of nonlinear eigenvector}

We prove here Proposition~\ref{prop:eigenvector}.

\medskip 

\proof 
For $t_0\in \cT$ fixed, dropped from the notation, {\bf H$_3$} yields,
$$
|\ffi([\omega])\ket\bra \ffi( [\omega])|\omega\ket =\omega.
$$
This requires $\omega$ to be parallel to $\ffi([\omega])$ where the latter is normalised.
We use Schauder's fixed point Theorem in a Banach space to actually prove that, locally, there exists $\omega$ such that
$\omega=\ffi([\omega])$, and thus $\|\omega\|=1$. 
Set  $B_1(\cH)=\{v\in \cH \ | \ \|v\|\leq 1\}$ and  {$S: B_1(\cH)\mapsto B_1(\cH)$} by  {$S(v)=\ffi([v])$}. 
This map is well defined, continuous and $B_1(\cH)$ is closed, convex and nonempty. Thus  {$S$} will have a fixed point if  {$\overline{S(B_1(\cH))}$} is compact. Let $K_\ffi=\{\ffi([x]) \, | \,x\in {[0,1]}^p\}$. By continuity of $\ffi$ in the variable~$x$, and compactness of ${[0,1]}^p$, $K_\ffi$ is compact.
Thus the closed subset  {$\overline{S(B_1(\cH))}$} of $K_\ffi$ is compact and Schauder Theorem   {(see~\cite{Evans} or Theorem 2.9 in \cite{LeDretbook}})
implies the existence of a fixed point for  {$S$}, for each given value of $t_0$. Since $\|\ffi([v])\|\equiv 1$, the normalization of the fixed point $\omega(t_0)$ holds. 

\medskip

In order to prove the smoothness of the map $\cT_0 \ni t\mapsto \omega(t)$, we use the implicit function theorem on the ${\mathcal C}^\infty$ map $J: \cT\times \cH\times \cH\ra \cH\times \cH$ defined by
$$
 J(t,v, \bar v)=\begin{pmatrix} J_1(t,v, \bar v) \cr J_2(t,v, \bar v) \end{pmatrix} 
 =\begin{pmatrix}
 v-\ffi(t,[v]) \cr \bar v -\bar \ffi(t,[v])
 \end{pmatrix}.
 $$
The zeros of $J$ define $\omega(t)$, in a neighbourhood of  $(t_0, \omega(t_0))$. Note that by a smooth change of phase we can consider locally the continuous vector $\ffi(t,x)$ defined by (\ref{defffiloc}).
For $1\leq j\leq p$, we compute, with $\{e_j\}_{j\in\mathbb N}$ the chosen orthonormal basis of $\cH$,
\begin{align}
&\partial_{v_j}J_1(t,v, \bar v)=e_j-\partial_{x_j}\ffi(t,[v])\bar v_j, \ \ \partial_{\bar v_j}J_1(t,v, \bar v)=-\partial_{x_j}\ffi(t,[v]) v_j\nonumber\\
&\partial_{v_j}J_2(t,v, \bar v)=-\partial_{x_j}\bar \ffi(t,[v])\bar v_j, \hspace{.8cm} \partial_{\bar v_j}J_2(t,v, \bar v)=e_j-\partial_{x_j}\bar \ffi(t,[v]) v_j.
\end{align}
Therefore, using the notation $D_{v,\bar v}J(t,v,\bar v)\in \cL(\cH\times \cH)$ for the 
derivative with respect to the variables $(v, \bar v)\in \cH\times \cH$, we get 
\be\label{derivatives}
D_{v,\bar v}J(t,v,\bar v)\begin{pmatrix} h \cr \bar h \end{pmatrix} = \begin{pmatrix} h \cr \bar h \end{pmatrix} - 
\sum_{j=1}^p \begin{pmatrix}  \partial_{x_j}\ffi(t,[v])\bra v_j e_j | h\ket +  \partial_{x_j} \ffi(t,[v])\bra \bar v_j e_j | \bar h\ket
\cr 
\partial_{x_j}\bar \ffi(t,[v])\bra v_j e_j | h\ket +  \partial_{x_j} \bar \ffi(t,[v])\bra \bar v_j e_j | \bar h\ket
 \end{pmatrix}.
\ee
We recall the notation in the scalar case 
\begin{align}
\nonumber
&\partial_z=\frac12(\partial_x-i\partial_y), \ \ \partial_{\bar z}=\frac12(\partial_x+i\partial_y)
\end{align}
such that if 
$f(x,y)\equiv g(z,\bar z)$ with $z=x+iy\in \C$ and $t=h+ik\in \C$, then
$$
Df(x,y)(h,k)=\partial_xf(x,y) h+\partial_yf(x,y) k\equiv D_{z,\bar z}g(z,\bar z)(h+ik,h-ik)=\partial_z g(z, \bar z) t+ \partial_{\bar z} g(z, \bar z) \bar t.
$$
With these notations in mind, we obtain equivalently
\begin{equation*}
D_{v,\bar v}J(t,v,\bar v) = {\rm Id} -
\sum_{j=1}^p \left| \begin{pmatrix}  \partial_{x_j}\ffi(t,[v])
\cr 
\partial_{x_j}\bar \ffi(t,[v])
 \end{pmatrix} \right\rangle 
 \left\langle \begin{pmatrix}   v_j e_j 
\cr 
 \bar v_j e_j 
 \end{pmatrix} \right|.
\end{equation*}
Therefore, for $v\in B_1(\cH)$, it is enough to show that $\|\partial_{x_j}\ffi(t,x)\|<1/4$, say, to satisfy the assumptions of the implicit function theorem. We compute
\begin{align}
\partial_{x_j}\ffi(t, x)=\frac{\partial_{x_j}P(t,x)\ffi(t_0,x_0)- P(t,x)\ffi(t_0,x_0) \bra \ffi(t_0,x_0)| \partial_{x_j}P(t,x) \ffi(t_0,x_0)\ket}{\bra \ffi(t_0,x_0)| P(t,x) \ffi(t_0,x_0)\ket^2},
\end{align} 
the norm of which is bounded above by $8\delta/g$, in a neighbourhood of $(t_0, x_0)$ characterised by 
$$\|P(t,x)\ffi(t_0,x_0)\|\geq 2^{-1/4}.$$
Hence, {\bf H$_1$} with $\delta $ small enough yields the existence of an open neighbourhood $\cT_0\ni t_0$ and of a ${\mathcal C}^\infty$ map $t\mapsto \tilde \omega(t)$ with $\tilde\omega(t_0)=\omega(t_0)$ that is solution to~(\ref{defomega}) for all $ t\in \cT_0$. 

\medskip 

To conclude, the proof, we observe that the argument above ensures $\|\tilde\omega(t)\|\equiv 1$, so that the phase adjustment
$\omega(t)= \tilde \omega(t)e^{-\int_{t_0}^t\bra \tilde\omega(s)| \dot {\tilde\omega}(s)\ket ds}$
implies that $\omega(t)$ satisfies $\bra \omega(t)| \dot {\omega}(t)\ket\equiv 0.$
\qed

\begin{rem}\label{rem:unbounded}
Note that in the proof above, we have not used the assumption $H(t,x)\in{\mathcal L}({\mathcal H})$ so that the result of Proposition~\ref{prop:eigenvector} extends to unbounded families of operators $H(t,x)$. 
\end{rem}

\subsection{Failure of global nonlinear eigenvectors}

We illustrate with  an  example the fact that the eigenvector constructed in Proposition~\ref{prop:eigenvector} may only exists locally.  For this, we consider the matrix-valued case where $H(t,x)$ is the real, symmetric, traceless two by two matrix 
$$ H(t,x)=\begin{pmatrix}
\cos (t\theta(x) ) & \sin (t\theta(x))\\
\sin (t\theta(x)) & -\cos (t \theta(x))
\end{pmatrix},$$
where $\R\ni x\mapsto \theta(x)$ is ${\mathcal C}^\infty$ and will be  chosen later. 
The eigenvalues of $H(t,x)$ are~$+1$ and~$-1$ with associated eigenvectors 
$$V_+(t,x)= \begin{pmatrix} \cos \big({t \theta(x)\over 2}\big)  \\ \sin \big({t \theta ( x)\over 2}\big)\end{pmatrix},\;\; V_-(t,x)= \begin{pmatrix}-\sin\big({t \theta(x)\over 2}\big)\\ \cos \big(t{\theta(x)\over 2}\big)  \end{pmatrix},$$
respectively. We denote by $P(t,x)$ the eigenprojectors for the eigenvalue $+1$. Then a real normalised vector $\omega(t)=\begin{pmatrix}\omega_1(t)\\ \omega_2(t)\end{pmatrix}$ satisfies
$$P(t,[\omega(t)] )\omega(t)= \omega(t)$$ 
if and only if 
$$\omega_1(t) = \cos \left({ t\over 2}\, \theta( |\omega_1(t)|^2) \right),\;\; \omega_2(t) =\sin \left({t\over 2}\, \theta(  |\omega_1(t)|^2) \right),$$
up to a global sign.
It is then enough to find the function $t\mapsto \omega_1(t)$. For fixed $t$, it reduces to finding $Y\in[0,1]$ such that
\begin{equation}\label{eq:omega}
 Y= \cos \left({t\over 2} \, \theta(Y^2)\right).
 \end{equation}
Let us restrict to $t\in [0,1]$ and choose  the function $\theta$ according to the following picture

\medskip 

\centerline{
\begin{tikzpicture}
   \tkzTabInit{$y$ / 1 , $\theta(y^2)$ / 1.5}{$0$, $y_{\rm max}$,  $y_{\rm min}$, $1$}
   \tkzTabVar{-/ $0$, +/ $\theta_{\rm max}$,-/  $\theta_{min}$,  +/ ${\pi}$}
\end{tikzpicture}
}

\medskip

\noindent We  fix $\theta_{\rm max} < {\pi}$  so that $\cos {{t\over 2} \theta_{\rm max}} >0$ for all $t\in [0,1]$ and ${\rm cos}$ is decreasing on the set of values of~${t\over 2}\theta$.
For $t=0$, the uniqueness of the solution of the equation~(\ref{eq:omega}) is guaranteed and for $t\in(0,1]$, it depends on whether 
$\cos {{t\over 2}  \theta_{\rm max}} <y_{\rm max}$ or not. Therefore, if 
we choose $y_{\rm max}$ and $\theta_{max}$ such that 
$$\cos {  {1\over 2} \theta_{\rm max} } <y_{\rm max},$$
we know that there exists $\tau\in ]0,1[$ such that the equation~(\ref{eq:omega}) has a unique solution for $t\in[0,\tau)$ and exactly three solutions for times $t\in (\tau,1]$. Figure~1 illustrates that fact.
 {Hence a solution chosen in a neighbourhood of $y_{\rm max}$ for times $t > \tau $ will disappear as $t$ passes the value $\tau$. }

\medskip

\begin{figure}
  \centering
  \begin{subfigure}[h]{.41\textwidth}
    \includegraphics[width=\textwidth]{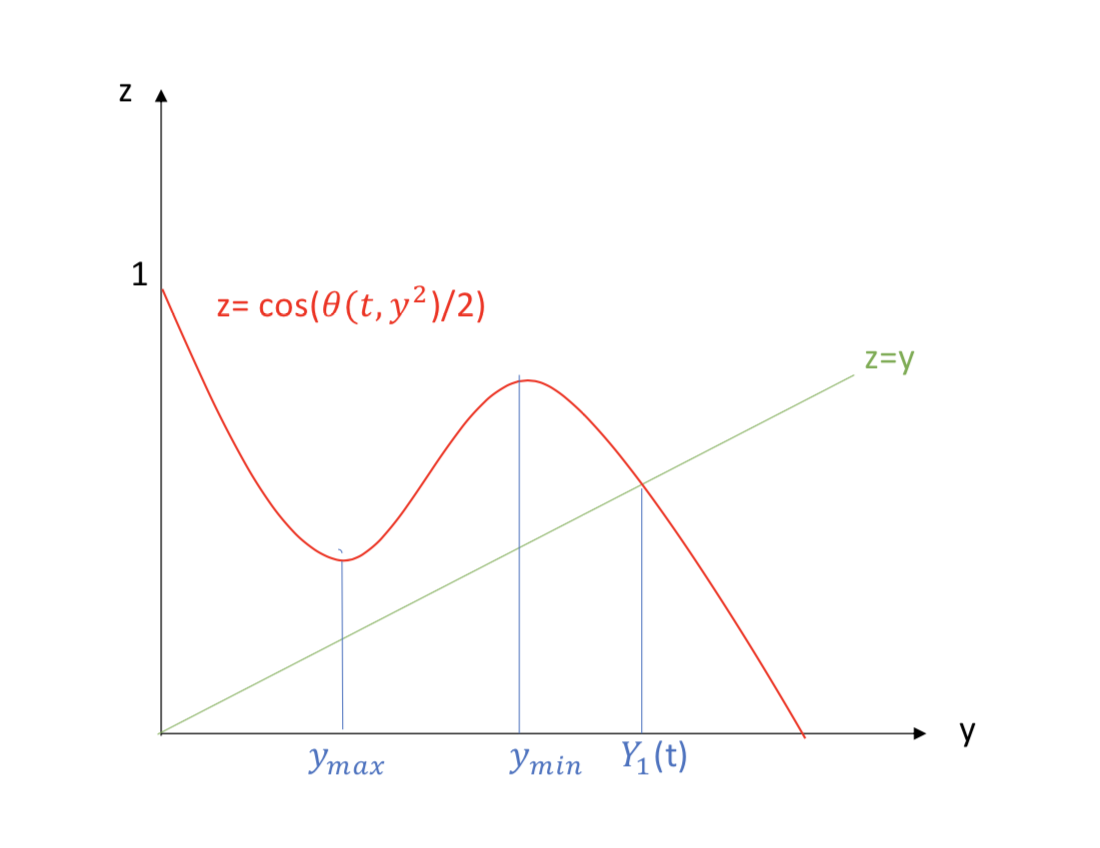}
   \caption{ $t<\tau$}
  \end{subfigure}
  \begin{subfigure}[h]{.40\textwidth}
    \includegraphics[width=\textwidth]{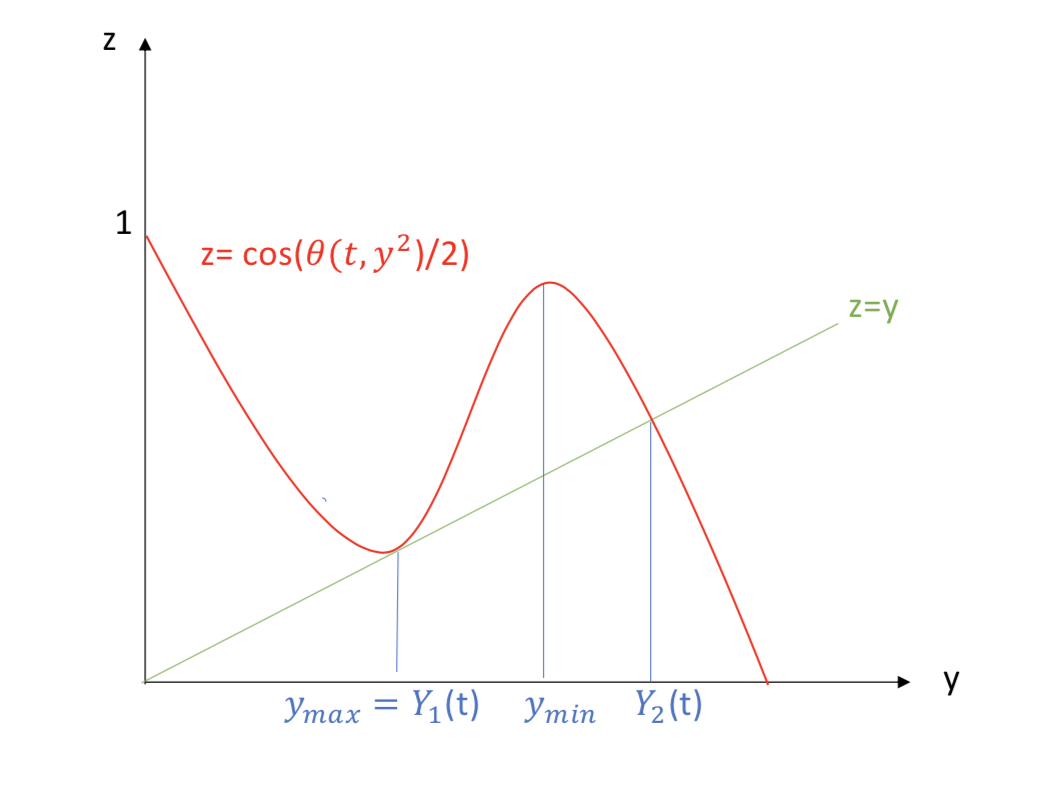}
     \caption{$t=\tau$ }
  \end{subfigure}
  \begin{subfigure}[h]{.40\textwidth}
    \includegraphics[width=\textwidth]{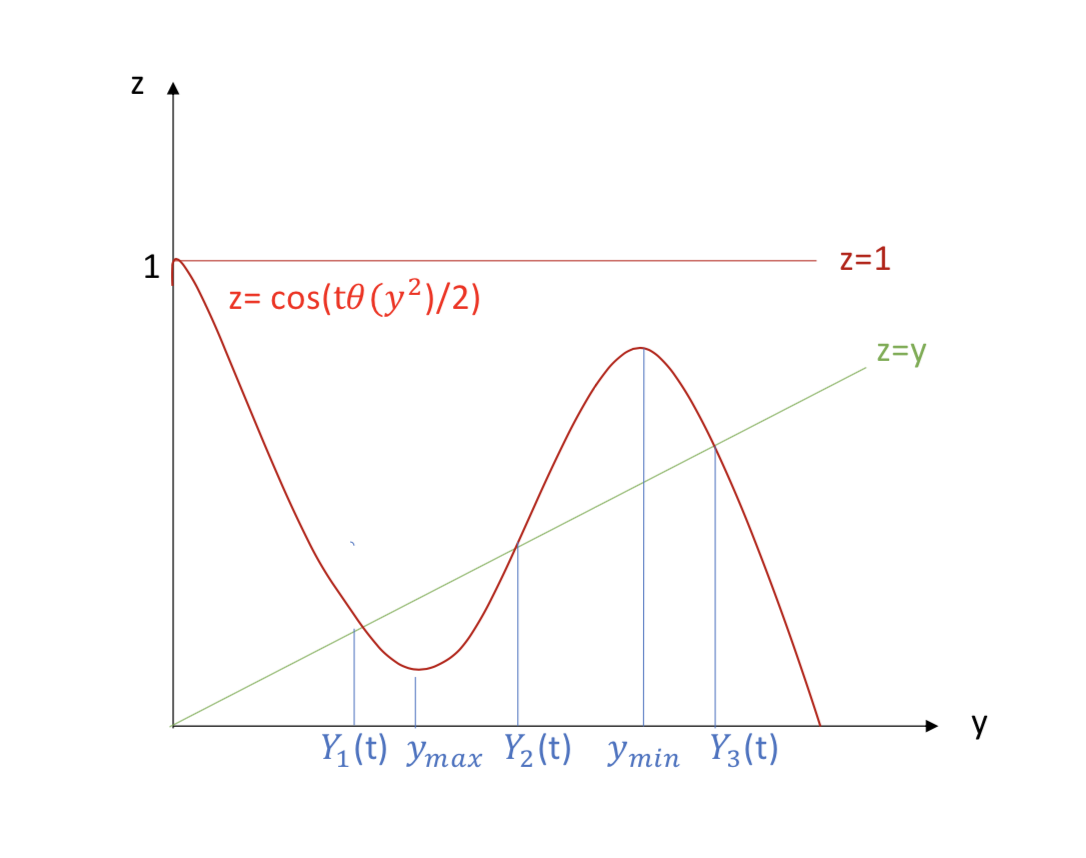}
  \caption{$t>\tau$}
  \end{subfigure}

 \caption{\small  Non uniqueness and degeneracy of the solutions of Equation~(\ref{eq:omega}): these three situations illustrating the dependance of the number of solutions depending  on the value of  
 $\cos {  {t\over 2} \theta_{\rm max} } $ comparatively with $y_{\rm max}$.}
\end{figure}


\section{ {The case of bounded operators}}\label{sec:proofmain} 

In this section, we prove Theorems~\ref{thm:main0} and~\ref{thm:main} the proofs of which both follow the same scheme. We first give the plan, spelling out the main steps and  lemmas  that we then successively prove in the next sections. 

\subsection{Proof of Theorems \ref{thm:main0} and~\ref{thm:main}}

Thanks to Lemma \ref{lem:elem} with $\chi(t,x)=-\lambda(t,x)$, we can reduce the analysis to the case $\lambda(t,x)=0$ without loss of generality, by considering the shift
\be\label{shiftham}
H(t,x)\mapsto H(t,x)-\lambda(t,x). 
\ee
The eigenvalues of the operator $H(t,x)$ are then all shifted by $\lambda(t,x)$ and we denote them by 
$$0 \ \mbox{and} \ \lambda_j(t,x), \;j\in\N^*,$$
where the functions $\lambda_j(t,x)$ may have changed compared to what they were in the introduction.
We set 
$\Delta(t)=v^\eps(t) -  \omega(t) $.
Then, the map $t\mapsto \Delta(t)$ (which also depends on $\eps$) satisfies the system 
$$i\eps \dot \Delta(t)= H(t,[v^\eps(t)]) v^\eps(t)-i\eps \dot \omega(t),\;\;\Delta(0)=0,$$ 
using a dot  to express derivatives with respect to time. 
For all $t\in \cT_0$, the interval $\cT_0$ is the set of times around~$t_0=0$ where $\omega(t)$ given in Proposition \ref{prop:eigenvector} exists, we have
$$H(t,[v^\eps(t)])= H(t, [\omega(t)+\Delta(t)] )= H(t,[\omega(t)]) + 2 \sum_{j=1}^p \partial_{x_j} H(t,[\omega(t)]) \Re  (\overline \omega_j \Delta_j)+O(\| \Delta\|^2),$$
and using $H(t,[\omega(t)])\omega(t)\equiv 0$, we obtain 
$$i\eps \dot \Delta (t)= -i\eps \dot \omega(t) +H(t,[\omega(t)]) \Delta(t) + 2 \sum_{j=1}^p \partial_{x_j} H(t,[\omega(t)]) \Re  (\overline \omega_j \Delta_j) \omega(t) +O(\| \Delta\|^2).$$
The equation involves a source term, $-i\eps \dot \omega(t) $, and its linear part depends on $\Delta(t)$ and $\overline\Delta(t)$. 
We write it as a  system for these two vectors:
$$
\left\{
\begin{array}{rcl}
i\eps \dot \Delta (t)&=& -i\eps \dot \omega(t) +H(t,[\omega(t)]) \Delta(t) +  \sum_{j=1}^p \partial_{x_j} H(t,[\omega(t)])  (\overline \omega_j \Delta_j + \omega_j \overline \Delta_j  ) \omega(t)  +O(\| \Delta\|^2),\\
i\eps \dot {\overline\Delta} (t)&=& -i\eps \dot{\overline \omega}(t) -\overline H(t,[\omega(t)]) \overline\Delta(t) -  \sum_{j=1}^p \partial_{x_j} \overline H(t,[\omega(t)]) (\overline \omega_j \Delta_j + \omega_j \overline \Delta_j ) \overline\omega(t)  +O(\| \Delta\|^2).
\end{array}
\right.
$$
We set for $j\in\{1,\cdots ,p\}$,
\begin{equation}\label{def:vj}
v_j(t)= \partial_{x_j} H(t,[\omega(t)])\omega(t),
\end{equation}
and, for later purposes, 
we notice that it follows from $P(t,x)H(t,x)\equiv 0$ that 
 {
$$(\partial_{x_j}P(t,x))H(t,x)+P(t,x)(\partial_{x_j}H(t,x))\equiv 0$$
 and, together with $H(t,[\omega(t)])\omega(t)\equiv 0$,}
 we get 
\be\label{vjperp}
v_j(t)=
(\un - P(t,[\omega(t)]))v_j(t).
\ee
We also set, for $j\in\{1,\cdots ,p\}$,
 $$\mu_j= \begin{pmatrix} v_j \\ -\overline v_j \end{pmatrix} \;\;{\rm and} \;\;  \nu_j= \begin{pmatrix} \omega_j e_j \\ \overline \omega_j e_j  \end{pmatrix}$$
and rewrite the system as~\eqref{eq:system}, namely
$$i\eps 
\begin{pmatrix} \dot\Delta \\ \dot {\overline \Delta}\end{pmatrix}
= 
-i\eps \begin{pmatrix} \dot \omega \\ \dot {\overline \omega}\end{pmatrix}
 + F(t) \begin{pmatrix}   \Delta \\ {\overline \Delta}\end{pmatrix}
 +\begin{pmatrix}  r^\eps \\ -\overline r^\eps \end{pmatrix} ,\;\;
 \Delta(0)=0,
 $$
 with $ r^\eps(t) =O(\| \Delta(t)\|^2)$ and
  {
 \begin{equation}
 \label{defF}
 F(t)= F_0(t)+G(t)
 \end{equation}
 with 
 \[
 F_0(t)= \begin{pmatrix}  H(t,[\omega(t)]) & 0 \\ 0 & - \overline H(t,[\omega(t)]) \end{pmatrix},
 \]
 and }
 \begin{align}
 \label{def:G}
 G(t)&=\sum_{j=1}^p 
 \begin{pmatrix}
 \overline \omega_j(t) |v_j(t)\rangle \langle e_j| &  \omega_j(t)  v_j(t)\rangle \langle e_j|\\
 -\overline \omega_j(t) |\overline v_j(t)\rangle \langle e_j |&-  \omega_j(t) |\overline v_j(t)\rangle \langle e_j|
 \end{pmatrix}
= \sum_{j=1}^p    \left| \begin{matrix} v_j(t) \\ -\overline v_j (t)\end{matrix} \right\rangle 
 \left\langle \begin{matrix} \omega_j (t)e_j \\ \overline \omega_j(t)e_j \end{matrix} \right|  
 \\
 \nonumber
 & = \sum_{j=1}^p  |\mu_j(t) \rangle \langle \nu_j(t) |.
 \end{align}
Note that $F(t), G(t)\in {\mathcal L}({\mathcal H}\times {\mathcal H})$, $F(t)$  is non self-adjoint and $G(t)$ is of finite rank.  {Hence},   $G(t)$ can be treated as a perturbation of the self-adjoint
operator $F_0(t)$. One then observes that two classical consequences of Weinstein-Aronszajn formula are that  $\sigma_{ess}(F(t))=\sigma_{ess}(F_0(t))$, and  that $\sigma_d(F(t))$ consists in finitely many of eigenvalues (see e.g. \cite{K2}, Chap. IV, $\S$~6).

\medskip 

The structure of the spectrum of $F(t)$ is crucial for our analysis. As we shall see in the following, the proof of Theorems~\ref{thm:main0} and~\ref{thm:main} works out when the spectrum of $F(t)$ is semisimple with real eigenvalues of constant multiplicity for all $t\in{\mathcal T}_0$. Moreover, there are situations where this can be proved and the next lemma describes such cases.  According to the assumptions of Theorems~\ref{thm:main0} and~\ref{thm:main}, we  focus on the case where $H(t,[\omega(t)])$ is real.

\begin{lem}\label{lem:specft2}
a) There exists $\delta_0>0$ such that if, for all 
 $t\in \cT_0$,  we have ${\bf H}_0$, ${\bf H}_1$ for some $\delta<\delta_0$, ${\bf H}_2$ and ${\bf H}_3$, then $0\in \sigma(F(t))$ as a doubly degenerate isolated eigenvalue, with corresponding eigennilpotent ${\mathbb N}_0(t)\equiv 0$. \\
b)  Moreover, if $H(t,x)=\overline{H}(t,x)$, $\sigma(H(t,x))$ is simple and $\sigma(H(t,x))\cap \sigma(-H(t,x))=\{0\}$
 for all $(t,x)\in \cT_0\times \cX^p$, then $\delta_0$ can be chosen so that 
the spectrum of $F(t)$ is real-valued for all $t \in \cT_0$ and takes the form 
\be \label{simplespect}
-\ell_{N-1}(t) <\cdots <-\ell_1(t)<0<\ell_1(t)<\cdots< \ell_{N-1}(t),
\ee
where $\ell_0(t)\equiv 0$ is of multiplicity two, and each eigenvalue $\pm \ell_k(t)$, $1\leq k\leq N-1$ is simple.\\
c) Finally, in the special case $p=1$, we have a) and if moreover $H(t,x)=\overline{H}(t,x)$ and  $\sigma(F(t))\setminus{\sigma(F_0(t))}$ consists in exactly
 $2(N-1)$ perturbed eigenvalues, then $\sigma(F(t))\subset \R$ and all corresponding eigennilpotents are zero. 
\end{lem}

 {Recall that the eigennilpotents correspond to the Jordan blocks in finite dimension.}
The points~a) and~b) imply that under the assumptions of Theorem~\ref{thm:main0}, the spectrum of $F(t)$ is semisimple with real eigenvalues of constant multiplicity for all $t\in{\mathcal T}_0$, and thus that the assumptions of Theorem~\ref{thm:main} are satisfied. The point c) gives another situation with possibly degenerate eigenvalues where the assumptions of Theorem~\ref{thm:main} hold.

\begin{rem} \label{rem:gen}
i) Note that for b), it is enough to assume $\sigma(H(0,x))\cap \sigma(-H(0,x))=\{0\}$, at the cost of making $|\cT_0|$ smaller. This is a generic hypothesis which automatically satisfied whenever $\lambda_{j_0}$ is the ground state or the upper eigenvalue. \\
ii)  The condition $\#\{\sigma(F(t))\setminus \sigma(F_0(t))\}=2(N-1)$ states that the spectral effect of the rank one perturbation $G(t)$ is maximal, which is a genericity assumption.%
 The multiplicities of the eigenvalues of $F_0(t)$ are arbitrary, possibly infinite, so that case c) does not necessary reduce to finite dimension, in contrast to the situation dealt with in Theorem~\ref{thm:main0}.  \\
iii) Besides,  if the spectral effect of the rank one perturbation is maximal on $\cT_0$, then $\sigma(F(t))$ takes the form (\ref{simplespect})  for all $t\in \cT_0$, with 
$4(N-1)$ non zero distinct eigenvalues instead of $2(N-1)$, $2(N-1)$ of which are simple  {(provided all eigenvalues of $F_0(t)$ have multiplicity at least two)}.\\
iv) The condition $H(t,x)$ real does not seem strong enough to ensure $\sigma(F(t))\subset {\mathbb R}$ for $p\geq 2$; 
see the example of the Hamiltonian given by equation~(\ref{ex2}) in Appendix {\em A}.
\end{rem}

We 
 postpone the proof of Lemma~\ref{lem:specft2}  {to Section~\ref{sec:32}} and we go to the next step of the proof which consists in controlling the adiabatic limit of the two-parameter evolution operator $T^\eps(t,s)$ generated by $F(t)$ (see~(\ref{semiF}) below), and using it to estimate $\left\|  \begin{pmatrix}   \Delta \\ {\overline \Delta}\end{pmatrix}\right\|$ via Duhamel formula. Since $F(t)$ is not self-adjoint, this requires some care because the possible occurence of nilpotent operators in its spectral decomposition leads to subexponential divergence of the semigroup  as $\eps\ra 0$ (see~\cite{J2}), that we cannot accommodate.  However,  Lemma \ref{lem:specft2} ensures that 
under the assumptions of Theorem~\ref{thm:main0}, and by hypothesis in Theorem~\ref{thm:main},  for all $t\in\cT_0$, $F(t)$ is semi-simple, with spectral decomposition
\be\label{decspec}
F(t)=\sum_{j=-N'}^{N'} \ell_j(t) {\mathbb P}_j(t), \ \mbox{with the convention} \ \ell_{-|j|}(t)=-\ell_{|j|}(t)  {<0},
\ee
where we have set  $N'=N-1$ for convenience  and where $ {\mathbb P}_j(t)$ are ${\mathcal C}^\infty$ eigenprojectors corresponding to real eigenvalues $\ell_j(t)$. We now work under these assumptions.

\medskip 

Despite the eigenprojectors ${\mathbb P}_j(t)$ not being orthogonal, with norms possibly larger than~$1$,  we prove in the next lemma that any operator 
$F(t)$ with real spectrum satisfying~(\ref{decspec}) generates an evolution operator which is uniformly bounded in $\eps$ and almost intertwines its eigenprojectors in the adiabatic limit,  {in the sense of (\ref{commutation}) below}. In line with Kato's approach (\cite{K1} and {\it e.g.}~\cite{HJ}), we introduce the dynamical phase operator~$\Phi^\eps(t,s)$ defined by 
\be\label{dynpha}
\Phi^\eps(t,s)= \sum_{j=-N'}^{N'} {\mathbb P}_j(0) {\rm e}^{-{i\over \eps}\int_s^t \ell_j(\sigma) d\sigma}, \ \mbox{s.t.} \ \Phi^\eps(t,s)^{-1}=\Phi^\eps(s,t), 
\ee
and the intertwining operator $W(t)$ given by
\be\label{paraltrans}
i\partial_t W(t) = K(t) W(t) ,\;\;W(0)={\rm Id},\;\;
{\rm with}\;\; 
K(t)=i \sum_{j=-N'}^{N'} \dot {\mathbb P}_j(t){\mathbb P}_j(t) .
\ee
As is well known (see~\cite{K2}), for all $t\in\R$, we have 
\be\label{commutation}
{\mathbb P}_j(t)W(t) =W(t){\mathbb P}_j(0),
\ee
and thanks to Lemma \ref{lem:specft2}, $\|\Phi^\eps(t,s)\|$ is uniformly bounded in $\eps$. 
Moreover, we check that
\be\label{diffeqphi}
i\eps\partial_t \Phi^\eps(t,s)=W(t)^{-1}F(t)W(t)\Phi^\eps(t,s)\equiv \tilde F(t) \Phi^\eps(t,s).
\ee
We then introduce the bounded family of operators 
\begin{equation}\label{def:Veps}
V^\eps(t,s)=W(t) \Phi^\eps(t,s) W(s)^{-1},
\end{equation}
which satisfy $V^\eps(t,s)^{-1}=V^\eps(s,t)$ and 
\begin{align}\nonumber
&V^\eps(t,s){\mathbb P}_j(s)={\mathbb P}_j(t)V^\eps(t,s)= W(t) {\mathbb P}_j(0) {\rm e}^{-{i\over \eps}\int_s^t \ell_j(\sigma) d\sigma} W(s)^{-1}.
\end{align}
Moreover, because $F(t)$ is semi-simple, $V^\eps(t,s)$ approximates the evolution operator generated by $F(t)$, as described by the next lemma which applies in a quite general setting. 

\begin{lem}\label{lem:adiab} 
Let  $\cT$ be an open bounded interval of $\R$ containing $0$ and consider the operator defined on a Hilbert space $\cK$ for all
$(t,s)\in \cT\times \cT$  by the strong differential equation
\be\label{semiF}
i\eps\partial_t T^\eps(t,s)= F(t) T^\eps(t,s),\;\; T^\eps(s,s)={\rm Id}.
\ee
If $F\in C^2(\cT, \cL({\mathcal K}))$ with continuous derivatives at $\partial \cT$ and if $F(t)$ is semi-simple and satisfies (\ref{decspec}) for all $t\in\overline{\cT}$, then 
we have  in ${\mathcal L}( \cK)$, 
$$
 T^\eps(t,s)= V^\eps(t,s) +O_{t,s}(\eps),
$$   
which implies the uniform boundedness of the family of operators $(T^\eps(t,s))_{\eps>0}$.
\end{lem}

\begin{rem}
i) As a consequence, $T^\eps(t,s){\mathbb P}_j(s)-{\mathbb P}_j(t)T^\eps(t,s)=O_{t,s}(\eps)$.
\\
ii) Note that $N'$ in (\ref{decspec}) is independent of $t\in\cT$, the multiplicities of the eigenvalues of $F(t)$ are arbitrary, possibly infinite. 
\end{rem}

We postpone the proof of this lemma   {to Section~\ref{sec:adiab} below} and conclude the proof of  Theorem \ref{thm:main}. As already mentioned, Lemma~\ref{lem:specft2} ensures we can apply Lemma~\ref{lem:adiab}  to  $\cK={\mathcal H} \times{\mathcal H}$ and $\cT=\cT_0$
under the assumptions of Theorems~\ref{thm:main0} and~\ref{thm:main}.
We write 
\begin{eqnarray}\label{bestbound}
 \begin{pmatrix}   \Delta(t) \\ {\overline \Delta}(t)\end{pmatrix}
&=&-  \int_0^t T^\eps (t,s) \begin{pmatrix} \dot \omega(s) \\ \dot {\overline \omega}(s)\end{pmatrix} ds
-{i\over\eps} \int_0^t T^\eps (t,s)
 \begin{pmatrix} r^\eps(s) \\ -\overline r^\eps(s) \end{pmatrix}
 ds \\ \nonumber
&=&-  \int_0^t V^\eps (t,s) \begin{pmatrix} \dot \omega(s) \\ \dot {\overline \omega}(s)\end{pmatrix} ds
-{i\over\eps} \int_0^t T^\eps (t,s)
 \begin{pmatrix} r^\eps(s)  \\ -\overline r^\eps(s) \end{pmatrix}
 ds +O_t(\eps).
\end{eqnarray} 
It follows  the definition of $\omega(t)$ that ${\mathbb P}_0(t)\begin{pmatrix} \dot \omega(t) \\ \dot {\overline \omega}(t)\end{pmatrix}= 0$ 
for all time $t\in\cT_0$. Therefore,  a classical adiabatic 
 argument (that we spell out in Section~\ref{sec:adiab} below)  {yields that Lemma~\ref{lem:adiab} has the consequence stated below.}

\begin{cor}\label{lem:dotomega}
For all $t\in\cT_0$, we have 
 $$\int_0^t V^\eps (t,s) \begin{pmatrix} \dot \omega(s) \\ \dot {\overline \omega}(s)\end{pmatrix} ds =O_t(\eps).$$
\end{cor}

 Therefore, focusing on the first component of (\ref{bestbound}) and setting
 $$\delta_\tau^\eps=\sup_{t\in[0,\tau]} \|\Delta(t)\|,$$ 
with $\tau \in \cT_0$, we deduce from the above that there exist $a,b>0$ such that  
 $$\delta_\tau^\eps\leq \eps a+ {b\over \eps}\tau  \delta_\tau^\eps{} ^2.$$
 Setting $X_\eps(\tau)= \eps^{-1} \delta_\tau^\eps$, we are led to study of the second order equation 
 \be\label{smalldelta}
 b\tau X^2 -X+a\geq 0.
 \ee
 Since $X_\eps(0)=0$, we deduce that $X_\eps(\tau)\leq {1\over 2b\tau} \left(1-\sqrt{ 1-4ab\tau}\right)=2a/(1+\sqrt{ 1-4ab\tau})$, as long as $4ab\tau\leq 1$. Finally, we obtain 
 $$\forall \tau\in [0,{(4ab)}^{-1}]\cap \cT_0,\;\; \delta_\tau^\eps \leq 2a\eps.$$
To justify the estimate (\ref{balancebound0}) for $t$ small, we start from (\ref{bestbound}) to get the existence of $\alpha, \beta >0$ such that
$$
\delta_\tau^\eps\leq \alpha \tau +{\beta\over \eps}\tau  \delta_\tau^\eps{} ^2.
$$
Focusing on times $\tau\leq \eps$, we consider 
$
\delta_\tau^\eps\leq \alpha \tau +{\beta} \delta_\tau^\eps{} ^2,
$
which, by a similar argument using $\delta_0^\eps=0$, implies, as long as $4 \alpha\beta\tau \leq 1$,
$
\delta_\tau^\eps\leq 2\alpha \tau.
$
Increasing the constant $\alpha$  if necessary, we get (\ref{balancebound0}).
 \qed

\medskip 

 {The two next subsections are respectively devoted to the spectral analysis of $F(t)$ with the proof of Lemma~\ref{lem:specft2},  and to the non self-adjoint adiabatic estimates with the proof of Lemma~\ref{lem:adiab} and its Corollary~\ref{lem:dotomega}.}

\subsection{Spectral analysis of $F(t)$}\label{sec:32}

We proceed with the proof of Lemma \ref{lem:specft2}, which relies on a careful analysis of the eigenvalues of $F(t)$ and of their multiplicity.
 
 \medskip 
 
Recall that $C$ denotes the anti-unitary involution defined on ${\mathcal H}$ by $C\psi=\overline \psi$ for all $\psi\in{\mathcal H}$. 
It is at this stage of the proof that we shall use the assumption $H=CHC=\overline H$, which implies $\omega=C\omega=\overline \omega$ and $v_j=Cv_j=\overline v_j$ for all $1\leq j\leq p$. 
Due to  assumption {\bf H}$_1$, we consider the operator $F(t)$ as a perturbation of the bloc diagonal operator
$F_0(t)$. 
Hence, since $\sigma(\overline {H}(t,[\omega(t)]))=\sigma(H(t,[\omega(t)])) $, 
$$
\sigma(F_0(t))=\sigma(H(t,[\omega(t)]))\cup \sigma(- H(t,[\omega(t)])).
$$
By our
genericity assumption, and due to the reduction we have made to the case where $\lambda(t,[\omega(t)])\equiv0$, the spectrum of $F_0(t)$ consists of $2N-1=2N'+1$ isolated eigenvalues 
$$-|\lambda_{N'}(t,[\omega(t)])|<\cdots <-|\lambda_1(t,[\omega(t)])|<0<|\lambda_1(t,[\omega(t)])|<\cdots<|\lambda_{N'}(t,[\omega(t)])|.$$
Since the operator $G(t)$ is of small norm by assumption {\bf H}$_1$ and its definition (equations~(\ref{def:vj}) and~\ref{def:G})), the spectrum of $F(t)$ can be inferred from that of $F_0(t)$ by  perturbation theory.  Hence $F(t)$ has  eigenvalues located in small discs ${\mathcal B}_j^\pm$  centered at  $\pm\lambda_j(t,[\omega(t)])$ and in a disk ${\mathcal B}_0$ with center~$0$. One can assume that these disks are of same radius~$r>0$ and that they do not intersect. Besides
\begin{itemize} 
\item in~${\mathcal B}_j^\pm$, $F(t)$ has as many eigenvalues (counted with multiplicity) as the 
multiplicity of $\lambda_j(t,[\omega(t)])$ as an eigenvalue of $F_0(t)$, and in case the multiplicity is infinite, there are only finitely many eigenvalues of $F(t)$ in ${\mathcal B}_j^\pm$ that differ from $\lambda_j(t,[\omega(t)])$,
\item in ${\mathcal B}_0$, $F(t)$ has at most two eigenvalues (counted with multiplicity). 
\end{itemize}
We are going to use symmetry considerations to prove that these eigenvalues are real-valued and have the same symmetry properties as those of $F_0(t)$.

\begin{rem}\label{rem:appendixA}
 We develop in Appendix~A an argument showing that the spectrum of $F(t)$ is not necessarily real if $H(t,x)$ is real, in order to motivate the assumptions  that its eigenvalues are simple. 
\end{rem}

\proof
a) We start by considering the spectrum of $F(t)$ in a neighbourhood of zero. For any $z\in {\mathcal B}_0\setminus\{0\}$, we can write
\begin{align}
F(t)-z = (F_0(t)-z)\left[\un + (F_0(t)-z)^{-1}G(t)\right].
\end{align}
Introducing the spectral projector  $\tilde P_0(t)$ associated with the doubly degenerate eigenvalue zero of $F_0(t)$ and the corresponding reduced resolvent
acting on $\tilde Q_0(t)(\cH\times\cH)$,
 $\tilde Q_0(t)=\un -\tilde P_0(t)$, we have for $z\in {\mathcal B}_0\setminus\{0\}$,
\be\label{resolvanteF_0} 
(F_0(t)-z)^{-1}=-\frac{\tilde P_0(t)}{z}+(F_0(t)-z)^{-1}_{\tilde Q_0(t)},
\ee
where we denote by $A_{\tilde Q_0}$ the restriction of the operator $A$ to the range of ${\tilde Q_0}$. 
Since 
 {
$$
\tilde P_0(t)=\begin{pmatrix}
|\omega(t)\ket\bra \omega(t) | & 0 \cr
0 & |{\omega}(t)\ket \bra { \omega }(t)| 
\end{pmatrix}
$$
}
and $\bra\omega(t) | v_j(t)\ket\equiv 0$ for all $1\leq j\leq p$, see (\ref{vjperp}), we get $\tilde P_0(t)G(t)\equiv 0$ so that,
\be\label{pertresF}
(F(t)-z)^{-1}=\left[\un + (F_0(t)-z)^{-1}_{\tilde Q_0(t)}G(t)\right]^{-1}(F_0(t)-z)^{-1}.
\ee
Indeed, the reduced resolvent is analytic in $z\in {\mathcal B}_0$ and $\|G(t)\|=2\delta$, so for $\delta_0$ small enough, the square bracket is invertible. Therefore, the only singularity of the resolvent of $F(t)$ lies at $z=0$, which remains a doubly degenerate eigenvalue after perturbation. The corresponding spectral projector is 
\be\label{pertP_0}
{\mathbb P}_0(t)=-\frac{1}{2i\pi}\int_{\partial {\mathcal B}_0}(F(t)-z)^{-1}dz=\left[\un + F_0(t)_{\tilde Q_0(t)}^{-1}G(t)\right]^{-1}\tilde P_0(t).
\ee
and, in view of~(\ref{resolvanteF_0}) and~(\ref{pertresF}), the corresponding eigennilpotent ${\mathbb N}_0(t)=F(t){\mathbb P}_0(t)$ writes, (see~\cite{K2} Chapt. III,~$\S$5)
\begin{align}
{\mathbb N}_0(t)& 
=-\frac{1}{2i\pi}\int_{\partial {\mathcal B}_0}z(F(t)-z)^{-1} dz \\ \nonumber
&=-\frac{1}{2i\pi}\int_{\partial {\mathcal B}_0}\left[\un + (F_0(t)-z)^{-1}_{\tilde Q_0(t)}G(t)\right]^{-1}\left[-{\tilde P_0(t)}+(F_0(t)-z)^{-1}_{\tilde Q_0(t)}z\right] dz.
\end{align}
Since the integrand is analytic in $ {\mathcal B}_0$, we get that ${\mathbb N}_0(t)\equiv 0$, which ends the proof of a) of Lemma~\ref{lem:specft2}.

\medskip

b) The perturbation $G(t)$ being of finite rank, we compute the  Aronszajn-Weinstein determinant (\cite{K2}, p. 245) which reads in our case for all $z\in \rho(F_0(t))$, the resolvent set of $F_0(t)$,
\begin{align}\label{AWdet}
w(z)&=\det (\delta_{j,k} + \langle \nu_k(t)  |(F_0(t)-z)^{-1}\mu_j(t) \rangle )_{1\leq j,k \leq  p}\\ \nonumber
&=\det (\delta_{j,k} +  \bar \omega_k\langle e_k(t)  |(H(t)-z)^{-1}v_j(t) \rangle + \omega_k \langle (H(t)+\bar z)^{-1}v_j(t) | e_k(t) \rangle)_{1\leq j,k \leq  p}.
\end{align}
It follows that $w(z)=\overline{w(-\bar z)}$ for all $z\in {\mathcal B}_0$. Since the zeros of $w(z)$ yield the eigenvalues of $F(t)$ in $\rho(F_0(t))$, we obtain
$$
z\in \sigma(F(t))\cap \rho(F_0(t)) \Leftrightarrow -\bar z\in \sigma(F(t))\cap \rho(F_0(t)).
$$
Since $\overline{H} (t,x)=H(t,x)$, we deduce
\begin{align}\label{AWdetsym}
w(z)&=\det (\delta_{j,k} +\omega_k\langle e_k(t)  |((H(t)-z)^{-1}+(H(t)+z)^{-1})v_j(t) \rangle)_{1\leq j,k \leq  p}=w(-z).
\end{align}
It follows then that 
\be\label{symspec}
z\in \sigma(F(t))\cap \rho(F_0(t)) \Rightarrow \{z,\bar z, -z, -\bar z\}\in  \sigma(F(t))\cap \rho(F_0(t)).
\ee
The nonzero eigenvalues of $F_0(t)$ being simple by assumption, the same is true by perturbation theory for those of $F(t)$ and 
 (\ref{symspec}) shows they must be real. Moreover, these conclusions hold for any $t\in \cT_0$ under the stated hypotheses.  
 
 \medskip

c) We now assume $p=1$. Let $t\in\cT_0$ fixed and let us drop the time variable. We make use of~(\ref{specdec}), with a possible  relabelling of the eigenvalues due to the shift~(\ref{shiftham}), to write with $N'=N-1$
\be\nonumber
(H-z)^{-1}=\frac{P_0}{-z}+\sum_{j=1}^{N'}\frac{P_j }{\lambda_j-z}, \ \ \mbox{where}  \ \ \lambda_j\neq 0 \ \mbox{if}\ j\geq 1, \ \mbox{and} \  \lambda_0=0.
\ee
Thus, with $p=1$, $z\in \rho(F_0)$, and $P_0v_1=0$,
\be\label{AWpone}
w(z)=1+2\omega_1\sum_{j=1}^{N'} \frac{\lambda_j\bra e_1 | P_j v_1\ket}{\lambda_j^2-z^2}=\frac{\Pi_{k=1}^{N'}(\lambda_k^2-z^2)+
2\omega_1\sum_{j=1}^{N'} \Pi_{k=1 \atop k\neq j}^{N'}(\lambda_k^2-z^2)\lambda_j\bra e_1 | P_j v_1\ket}{\Pi_{k=1}^{N'} (\lambda_k^2-z^2)}.
\ee
The numerator is a polynomial of degree $2N'$ which, by assumption, possesses $2N'$ distinct simple roots in $\rho(F_0)$. These roots being in the neighbourhood of $\sigma(F_0)\setminus \{0\}$ for $\delta$ small, (\ref{symspec}) implies that they are real. This proves $\sigma (F)\subset \R$. 

\medskip

We now consider the eigennilpotents. The potentially nonzero eigennilpotents~${\mathbb N}_{\pm\lambda_j}$ are thus attached to the unperturbed eigenvalues $\pm \lambda_j$ with sufficient multiplicity, {\it i.e.}
with $\dim \tilde P_j\geq 3$
only. For $p=1$ and $z\not\in\R$, the resolvent takes the explicit form
\begin{align}\nonumber
(F-z)^{-1}&=\left\{\un - \frac{(F_0-z)^{-1}|\mu\ket\bra\nu | }{1+\bra \nu |(F_0-z)^{-1} \mu \ket }\right\}(F_0-z)^{-1}\\
\nonumber
&=\left\{\un - \frac{\omega_1}{w(z)}\sum_{j=1}^{N'}\left| \begin{matrix} \frac{P_jv_1}{\lambda_j-z} \\ \frac{P_jv_1}{\lambda_j+z}\end{matrix} \right\rangle 
 \left\langle \begin{matrix} e_1\phantom{\frac{1}{\lambda_j-z} } \\  e_1\phantom{\frac{1}{\lambda_j-z}}  \end{matrix} \hspace{-.8cm}\right | \right\}(F_0(t)-z)^{-1}.
\end{align}
The eigennilpotents are the coefficients, up to a sign, of the poles of order two of the resolvent at the eigenvalues. We consider the nonzero eigenvalue $\lambda_k$ only, $-\lambda_k$ being similar. Using the fact that the numerator $\tilde w(z)$ of $w(z)$ in (\ref{AWpone}) is nonzero at $\lambda_k$ by assumption, we have in a neighbourhood of $\lambda_k$
\be\nonumber
w(z)^{-1}=(\lambda_k-z)(\lambda_k+z)\Pi_{j\neq k}(\lambda_j^2-z^2)/\tilde w(z):=(\lambda_k-z)s_{k}(1+O(\lambda_k-z)),
\ee
with $s_k={2\lambda_k\Pi_{j\neq k}(\lambda_j^2-\lambda_k^2)}/{\tilde w(\lambda_k)}\neq 0$. Hence, for $z$ close to $\lambda_k$,
\begin{align}\nonumber
(F-z)^{-1}&=\left\{\un - \omega_1s_k\left| \begin{matrix} {P_kv_1} \\ 0 \end{matrix} \right\rangle 
 \left\langle \begin{matrix} e_1 \\  e_1 \end{matrix} \right | + O(z-\lambda_k) \right\}\begin{pmatrix}\frac{P_k}{\lambda_k-z} & 0 \\ 0 & 0\end{pmatrix}+ O(1)\\
 \nonumber
 &=\frac{1}{\lambda_k-z}\begin{pmatrix}{P_k}(\un- \omega_1s_k
 | v_1\ket \bra e_1| )P_k & 0 \\ 0 & 0\end{pmatrix}+ O(1).
\end{align}
The absence of pole of order two shows that ${\mathbb N}_{k}=0$, 
 and the computation above further yields 
\be\nonumber
{\mathbb P}_{k}=\begin{pmatrix}{P_k}(\un- \omega_1s_k
 | v_1\ket \bra e_1| )P_k & 0 \\ 0 & 0\end{pmatrix},
\ee
which concludes the proof.

\medskip 

We end the argument by briefly checking that ${\mathbb P}_{k}$ is a projector on ${\mathcal H}\times{\mathcal H}$, or equivalently that  ${P_k}- \omega_1s_k
 P_k| v_1\ket \bra e_1| )P_k$ is a projector on $\cH$. Since $P_k^2=P_k$
 $$
( {P_k}- \omega_1s_k
 P_k| v_1\ket \bra e_1| )P_k)^2={P_k}- 2\omega_1s_k P_k| v_1\ket \bra e_1| )P_k+\omega_1^2s_k^2
 P_k| v_1\ket \bra e_1| P_kv_1\ket \bra e_1| P_k.
 $$
 The right hand side equals  ${P_k}- \omega_1s_k P_k| v_1\ket \bra e_1| )P_k$ if
$ \omega_1s_k\bra e_1| P_kv_1\ket =1.$
 With the definition of $s_k$, this is equivalent to
$ 2\omega_1\lambda_k \Pi_{j\neq k}(\lambda_j^2-\lambda_k^2)\bra e_1| P_kv_1\ket=\tilde w(\lambda_k).$
Now, (\ref{AWpone}) gives
$$
\tilde w (\lambda_k)=\Pi_{l=1}^{N'}(\lambda_l^2-\lambda_k^2)+
2\omega_1\sum_{j=1}^{N'} \Pi_{l=1 \atop l\neq j}^{N'}(\lambda_l^2-\lambda_k^2)\lambda_j\bra e_1 | P_j v_1\ket,
$$
where the first term equals zero, while the only non zero term in the sum corresponds to $j=k$. Altogether, 
$\tilde w (\lambda_k)=2\omega_1\Pi_{l=1 \atop l\neq k}^{N'}(\lambda_l^2-\lambda_k^2)\lambda_k\bra e_1 | P_k v_1\ket$ which 
yields the result. 
 \qed

\subsection{Non-selfadjoint adiabatic estimates}\label{sec:adiab}

We prove here Lemma~\ref{lem:adiab} in a way that naturally adapts to the unbounded setting that we shall consider in Section~\ref{sec:gen}. 

\medskip

\proof
We first note that by the definition of $V^\eps$ and $K$ (see~(\ref{def:Veps}) and~(\ref{paraltrans})), we have 
$$i\eps\partial_t V^\eps(t,s)=W(t) \sum_{j=-N'}^{N'}  \ell_j(t){\mathbb P}_j(0){\rm e} ^{-{i\over \eps} \int_s^t \ell_j(s')ds' }W(s)^{-1} + \eps K(t) V^\eps(t,s).$$
Using~(\ref{commutation}) and ${\mathbb P}_j(0)^2={\mathbb P}_j(0)$, we obtain 
\begin{eqnarray*}
i\eps\partial_t V^\eps(t,s)&=&  \sum_{j=-N'}^{N'} \ell_j(t) {\mathbb P}_j(t)W(t) {\mathbb P}_j(0){\rm e} ^{-{i\over \eps} \int_s^t \ell_j(s')ds' }
 W(s)^{-1} + \eps K(t) V^\eps(t,s) \\
 & = &F(t)  \sum_{j=-N'}^{N'}
  {\mathbb P}_j(t)
W(t) {\mathbb P}_j(0){\rm e} ^{-{i\over \eps} \int_s^t \ell_j(s')ds' } W(s)^{-1} + \eps K(t) V^\eps(t,s),
\end{eqnarray*}
whence 
\be\label{defadiabV}
i\eps\partial_t V^\eps(t,s)=F(t) V^\eps(t,s)+ \eps K(t) V^\eps (t,s).
\ee
We can now  compare $T^\eps(t,s)$ and $V^\eps(t,s)$. Let $\Omega^\eps (t,s)= V^\eps(t,s)^{-1} T^\eps(t,s)$, we have
\be\label{diffomeg}
i\partial_t \Omega^\eps(t,s)=-V^\eps (t,s)^{-1} K(t) T^\eps(t,s)= - (V^\eps(t,s)^{-1} K(t) V^\eps(t,s)) \, \Omega^\eps(t,s),
\ee
or, equivalently
\be\label{intomeg}
\Omega^\eps(t,s) ={\rm Id} +i\int_s^t V^\eps(t',s)^{-1} K(t') V^\eps(t',s)\Omega^\eps (t',s) dt'.
\ee
With the shorthand $\tilde K(t')=W^{-1}(t')K(t')W(t')$, we have 
$$
V^\eps(t',s)^{-1} K(t') V^\eps(t',s)=W(s)\Phi^\eps(s,t')\tilde K(t')\Phi^\eps(t',s)W^{-1}(s)
$$
and 
$${\mathbb P}_j(0) \tilde K(t') {\mathbb P}_k(0) = i (1-\delta_{j,k}){\mathbb P}_j(0) \tilde K(t') {\mathbb P}_k(0) .$$
Therefore, for any $j$, 
\be\label{shiftint}
 {\mathbb P}_j(s)\Omega(t,s)= {\mathbb P}_j(s)+iW(s)\int_s^t  {\mathbb P}_j(0)e^{{i\over\eps}\int_s^{t'}\ell_j(u)du}\tilde K(t')(\un - {\mathbb P}_j(0))\Phi^\eps(t',s)W^{-1}(s)\Omega^\eps (t',s) dt'.
\ee
Now, observe that,
\be\label{redeqdiff}
i\eps\partial_{t'} e^{{i\over\eps}\int_s^{t'}\ell_j(u)du}\Phi^\eps(t',s)=\tilde F_j(t')e^{{i\over\eps}\int_s^{t'}\ell_j(u)du}\Phi^\eps(t',s)
\ee
where
\be\label{defftilde}
\tilde F_j(t)=\sum_{k} {\mathbb P}_k(0)(\ell_k(t)-\ell_j(t))=\tilde F(t)-\ell_j(t)\un 
\ee
is invertible on $(\un-{\mathbb P}_j(0)\cH\times\cH$, with reduced resolvent we denote by
$$
\tilde R_j(t):=\tilde F_j^{-1}(t)|_{\un - {\mathbb P}_j(0)}=\sum_{k \atop k\neq j} {\mathbb P}_k(0)/(\ell_k(t)-\ell_j(t)).
$$
 Thus the integrand in (\ref{shiftint}) reads, using (\ref{diffomeg}) in the last step,
\begin{eqnarray*}
I:&=&e^{{i\over\eps} \int_s^{t'}\ell_j(u)du/}\tilde K(t')(\un - {\mathbb P}_j(0))\Phi^\eps(t',s)W^{-1}(s)\Omega^\eps (t',s)\\
&=&
\tilde K(t')\tilde R_j(t')\tilde F_j(t')e^{{i\over\eps}\int_s^{t'}\ell_j(u)du}\Phi^\eps(t',s)W^{-1}(s)\Omega^\eps (t',s)\\
&=&\tilde K(t')\tilde R_j(t')\{i\eps\partial_{t'} e^{{i\over\eps}\int_s^{t'}\ell_j(u)du}\Phi^\eps(t',s)\}W^{-1}(s)\Omega^\eps (t',s).
\end{eqnarray*}
We deduce 
\begin{eqnarray}
\nonumber 
I&=&i\eps\partial_{t'} \left\{\tilde K(t')\tilde R_j(t')e^{{i\over\eps}\int_s^{t'}\ell_j(u)du}\Phi^\eps(t',s)W^{-1}(s)\Omega^\eps (t',s)\right\}\nonumber\\
&&-i\eps\partial_{t'} \{\tilde K(t')\tilde R_j(t')\}e^{{i\over\eps}\int_s^{t'}\ell_j(u)du}\Phi^\eps(t',s)W^{-1}(s)\Omega^\eps (t',s)\nonumber\\
\label{ippformul}
&&+\eps \tilde K(t')\tilde R_j(t')e^{{i\over\eps}\int_s^{t'}\ell_j(u)du}\tilde K(t')\Phi^\eps(t',s)W^{-1}(s)\Omega^\eps(t',s).
\end{eqnarray}
Note that thanks to  our spectral hypothesis, we have 
$$
\sup_{t\in\overline{\mathcal T}} \{\| \tilde R_j(t) \|, \| \partial_t \tilde R_j(t) \| \}  <c_0
$$
for some constant $c_0$. 
We can thus integrate (\ref{shiftint}) by parts to get the existence of a constant $c>0$ (that may change from line to line below) such that for all $t,s\in\overline{\mathcal T}$
\begin{align}
\|  {\mathbb P}_j(s)\Omega(t,s)- {\mathbb P}_j(s)\| \leq c\, c_0  \,\eps ||| \Omega |||,
\end{align}
where $||| \Omega |||=\sup_{(s,t)\in \overline \cT}\|\Omega(t,s)\|$. Therefore, 
$$
\sup_{(s,t)\in \overline \cT}\| \Omega(t,s)-\un \|^2\leq c\, c_0^2  \,\eps^2 ||| \Omega |||^2\leq c \,\eps^2 \Big(||| \Omega -\un |||^2+ ||| \un |||^2\Big),
$$
from which we get the existence of  $\eps_0>0$, independent of $t$, such that $\eps <\eps_0$ implies 
$$ ||| \Omega -\un |||=O(\eps).$$
Hence we infer the sought for bounds
$$
\Omega^\eps(t,s)= V^\eps(t,s)^{-1} T^\eps(t,s)={\rm Id} + O_{t,s}(\eps), \ \ \mbox{and} \ \ T^\eps(t,s)=O_{t,s}(1).
$$
\qed

\medskip 

\noindent Let us now prove  {Corollary}~\ref{lem:dotomega}. 
\medskip 

 \proof
Set $\chi_\omega(s)=\begin{pmatrix} \dot \omega(s) \\ \dot {\overline \omega}(s)\end{pmatrix}$ and 
recall that
$$\tilde P_0(s)\chi_\omega(s)= \tilde P_0(s) \begin{pmatrix} \dot \omega(s) \\ \dot { \overline \omega}(s)\end{pmatrix}\equiv 0.$$
Therefore, the perturbed projector ${\mathbb P}_0(s)$ associated to the kernel of $F(s)$ given by (\ref{pertP_0}) satisfies
\begin{equation*}
{\mathbb P}_0(s)\chi_\omega(s)
=\left[\un + F_0(s)_{\tilde Q_0(s)}^{-1}G(s)\right]^{-1}\tilde P_0(s)\chi_\omega(s)
\equiv 0.
\end{equation*}
Hence, writing $\tilde F(s)=W^{-1}(s)F(s)W(s)$, we have
$$
V^\eps(t,s)\chi_\omega(s)=W(t)\Phi^\eps(t,s)(\un-{\mathbb P}_0(0))W^{-1}(s)\chi_\omega(s)=
W(t)\Phi^\eps(t,s)\tilde F(s) (\tilde F(s)^{-1}W^{-1}(s)\chi_\omega(s)),
$$
where $\tilde F(s)^{-1}$ is to be understood as the reduced resolvent of $\tilde F(s)$ acting on $(\un-{\mathbb P}_0(0))\cH\times \cH$.
Thanks to (\ref{diffeqphi}) we can rewrite
\begin{align}
\int_0^t V^\eps(t,s)\chi_\omega(s) ds&=-\i\eps W(t)  \int_0^t \partial_s  \{\Phi^\eps(t,s)  (\tilde F(s)^{-1}W^{-1}(s)\chi_\omega(s))\}ds\\ \nonumber
&=-\i\eps W(t)\{\Phi^\eps(t,s)  (\tilde F(s)^{-1}W^{-1}(s)\chi_\omega(s))\}|^t_0=O_t(\eps).
\end{align}
\qed


\section{Generalization to unbounded operators}\label{sec:gen}

In this section, we prove Theorem~\ref{thm:unbounded}.  {To start with, we focus on the existence of global weak solutions in Section~\ref{sec:existence}. Then, to deal with the adiabatic approximation, we follow the same scheme of proof than in Section~\ref{sec:proofmain}, analyzing the function $\Delta(t)=\psi^\eps(t)-\omega(t)$  that solves a system of the form~\ref{eq:system} but now in the weak sense (see~\eqref{weakdeldot} ). This is explained in Section~\ref{sec:proof(2)}. However,  due to the unboundedness of the operator~$H(t,x)$, several technical points have to be taken care of:
\begin{enumerate}
\item The existence of the propagator associated with the operator $F(t)$ (Section~\ref{sec:proF(t)}),
\item The analysis of the (unbounded) spectrum of $F(t)$ (Section~\ref{sec:specF(t)}) proving an extension of Lemma~\ref{lem:specft2} b) with an infinite number of eigenvalues.
\item The construction of the associated adiabatic approximate propagator and of its properties (Section~\ref{sec:adiaF(t)}). 
\end{enumerate}
We can then conclude the proof of the Theorem~\ref{thm:unbounded} in Section~\ref{sec:conclusion}. 
}

\subsection{Proof of Theorem~\ref{thm:unbounded}(1)}\label{sec:existence}
We prove  the existence of a unique global solution to the nonlinear Schr\"odinger equation (\ref{eq:schroabs}) in the weak sense, {\it i.e.}
for any $\chi\in \cD$, we have equation~(\ref{weakschroabs}), that is 
$$
i\eps\partial_t \bra \chi | \psi ^\eps (t)\ket=\bra (H_0 +W(t,[\psi^\eps(t)]) \chi | \psi ^\eps (t)\ket, \ \ \psi ^\eps (0)= \omega(0).
$$
We denote by ${\rm e}^{-i tH_0}$ the evolution group associated with $H_0$ which 
maps $\cD$ into $\cD$ and is differentiable on $\cD$ only. 
We first consider a solution of~(\ref{eq:schroabs}) as an integral solution, {\it i.e.} a { continuous} 
function $t\mapsto \psi^\eps(t)\in \cH$ such that 
\be\label{intsol}
\forall t\in \cT,\;\; \psi^\eps(t) = {\rm e}^{-{i\over \eps} t H_0} \omega(0) +{1\over i\eps} \int_0^t {\rm e}^{-{i\over \eps} (t-s) H_0}  W(s, [\psi^\eps(s)]) \psi^\eps(s)ds.
\ee
Indeed, 
such a $ \psi^\eps(t)$ satisfies (\ref{weakschroabs}) for all $\chi\in{\mathcal D}$.
Besides, if it does exist, we will  show that  the solution satisfies $\| \psi^\eps(t)\| =\| \omega(0)\|=1$. 

\medskip 

To construct $\psi^\eps(t)$, we  consider $M\geq 1$, $\tau>0$ such that 
$$1+\tau M \sup \| W \| \leq M,\;\;\mbox {and}\;\; \sup (\|W \| + 4M \| d_xW \|)\tau <1, $$  
 the ball $B(0,M)$ of $\cH$  and the map
 $\Phi: {\mathcal C}([0,\tau\eps], \cH)\ra {\mathcal C}([0,\tau\eps], \cH)$ 
 $$\Phi : v(t) \mapsto {\rm e}^{-{i\over \eps} t H_0}\omega(0) +{1\over i\eps} \int_0^t {\rm e}^{-{i\over \eps} (t-s) H_0}  W(s, [v(s)]) v(s) ds.$$
By the choice of $\tau$, $\Phi$ maps  ${\mathcal C}([0,\tau\eps], B(0,M))$ into itself. 
Besides,
$\Phi$ is a contraction:
\begin{align}
\Phi(v)(t)-\Phi(w)(t)&={1\over i\eps} \int_0^t {\rm e}^{-{i\over \eps} (t-s) H_0}  (W(s, [v(s)]) v(s)-W(s, [w(s)]) w(s)) ds \\ \nonumber
&={1\over i\eps} \int_0^t {\rm e}^{-{i\over \eps} (t-s) H_0}  (W(s, [v(s)]) (v(s)-w(s))-(W(s, [v(s)]) - W(s, [w(s)])w(s)) ds,
\end{align} 
hence, uniformly in $t\in\overline{\cT}$, 
$$\|\Phi(v)-\Phi(w)\|\leq  \sup (\|W \| + 4M \| d_xW \|)\tau \| v-w\|$$
with $\sup (\|W \| + 4M \| d_xW \|)\tau<1$.
Therefore, $\Phi$ has a unique fixed point $\psi^\eps(t)\in C([0,\tau\eps],B(0,M))$,  which is the unique integral solution of the equation~(\ref{eq:schroabs}) on $[0,\tau\eps]$. 

\medskip

Now, the vector $\ffi^\eps(t)= {\rm e}^{{i\over \eps} t H_0} \psi^\eps(t)$ satisfies $\forall t\in [0,\tau\eps]$,
\be\label{intrep}
\ffi^\eps(t) = \omega(0) +{1\over i\eps} \int_0^t {\rm e}^{{i\over \eps} s H_0}  W(s, [\psi^\eps(s)]) {\rm e}^{-{i\over \eps} s H_0} \ffi^\eps(s)ds,
\ee
where the integrand is continuous, so that strong differentiation with respect to time is allowed. 
Since the operator ${\rm e}^{{i\over \eps} t H_0}  W(t, [\psi^\eps(t)]) {\rm e}^{{-i\over \eps} t H_0} $
 is self-adjoint, one gets in the usual way that, 
$$
\forall t \in [0,\tau\eps ],\;\; \|\ffi^\eps(t)\|=\|\psi^\eps(t)\|\equiv 1.
$$

\medskip

Observe that the choice of $\tau$ only depends on $\| W\| , \| d_xW\| $ and $M$, and since $\|\psi^\eps(\tau\eps)\|=1$, we can reiterate the same argument on $[\tau\eps, 2\tau\eps]$ starting from the initial data $\psi^\eps(\tau\eps)$ instead of 
$\omega(0)$. One then constructs the unique normalised integral solution of~(\ref{eq:schroabs}) on $[\tau\eps, 2\tau\eps]$, so that $\|\psi^\eps(2\tau\eps)\|=1$. Iterating the process, we see that we have a unique global integral solution of the form (\ref{intsol}) to the equation~(\ref{eq:schroabs}).

\subsection{Preparation of the proof  of Theorem~\ref{thm:unbounded} (2)}  \label{sec:proof(2)}

At this point, we follow the same strategy as in Section~\ref{sec:proofmain}. Here again, 
the gauge invariance manifested in the conclusions of Lemma \ref{lem:elem} holds in this case as well. This allows us to consider
the replacement
$H(t,x)\mapsto H(t,x)-\lambda(t,x)\un$,
keeping the  notation $H(t,x)$ for the shifted Hamiltonian, which admits~$0$ in its spectrum and finitely many negative eigenvalues. 
 We set 
$\Delta(t)=\psi^\eps(t)- \omega(t),$
which solves a system similar to~(\ref{eq:system}), as we now check.
With the definitions  
\be\label{defvj2}
|e_j\rangle \langle e_j| \omega(t)=\omega_j(t) e_j ,\;\; v_j(t)= \partial_{x_j} H(t,[\omega(t)]) \omega(t),
\ee
and for all normalized $\chi\in \cD$, we have, using the smoothness of the bounded operator $W(t,x)$, 
\begin{align*}
\nonumber
i\eps \partial_t \bra\chi | \Delta(t)\ket=& \bra (H_0+W(t,[\psi^\eps(t)]))\chi |\psi^\eps(t)\ket -i\eps \bra \chi |\dot\omega(t)\ket \\
\nonumber
&=\bra (H_0+W(t,[\omega(t)]))\chi | \Delta\ket +2\sum_{j=1}^p\bra \chi | \big( \partial_{x_j}W(t,[\omega(t)])\Re (\omega_j \overline{\Delta_j(t)})\omega(t)\ket 
 \\
 \nonumber
 &\qquad -i\eps \bra \chi |\dot\omega(t)\ket + \bra \chi |r^\eps(t)\ket.
\end{align*}
where 
$r^\eps(t)$ is of order $\|\Delta(t)\|^2$. 
Indeed, it takes 
the form 
\begin{eqnarray*}
 r^\eps(t)
& =&  \sum_{1\leq j, k\leq p}( \Delta_j(t) \overline  \Delta_k(t)  b_{j,k} (t) + \Delta_j(t)  \Delta_k(t)  \tilde b_{j,k} (t) + \overline \Delta_j(t) \overline  \Delta_k(t)  \underline b_{j,k} (t))\\
\nonumber
&& 
  + \sum_{1\leq j\leq p} (\Delta_j(t) B_j(t) \Delta(t) +\overline \Delta_j(t) \tilde B_j(t) \Delta(t)
 \end{eqnarray*}
for some  uniformly bounded vectors $\tilde b_{j,k}(t), \underline b_{j,k}(t), b_{j,k}(t)\in{\mathcal H}$   and  uniformly bounded operators $B_j(t)$ and~$\tilde B_j(t)$  
(which may also depend on $\Delta(t)$ and $\overline\Delta(t)$):
\begin{eqnarray}\label{prop:reps}
 r^\eps(t)
& =& \sum_{1\leq j\leq p} \int_0^1 \partial_{x_j} W(t, [\omega +s\Delta]) (2 \Re (\omega_j\overline \Delta_j) \Delta + |\Delta_j|^2(\Delta+\omega)) ds \\
\nonumber
&&+  \sum_{1\leq j,k\leq p} \int_0^1 (1-s)(2 \Re (\omega_j\overline \Delta_j)) (2\Re (\omega_k\overline \Delta_k)+|\Delta_k|^2)ds\;
\partial^2_{x_j,x_k} W(t,[\omega+s\Delta])  \omega ds.
\end{eqnarray}
Besides,  $i\eps \partial_t \bra\chi | \overline{\Delta(t)}\ket$ satisfies a similar equation corresponding to (\ref{eq:system}). 
Thus for the nonlinear problem, we need to consider weak solutions on $\cD\times \cD$  of the coupled equations: 
For any $(\chi_1, \chi_2)\in \cD\times \cD$, 
\be\label{weakdeldot}
i\eps \partial_t \left\langle \begin{pmatrix} \chi_1  \\ \chi_2 \end{pmatrix} \Big| \begin{pmatrix} \Delta \\ {\overline \Delta}\end{pmatrix}\right\rangle
= 
-i\eps \left\langle \begin{pmatrix} \chi_1  \\ \chi_2 \end{pmatrix} \Big| \begin{pmatrix} \dot \omega \\ \dot { \omega}\end{pmatrix}\right\rangle
 + \left\langle F^*(t) \begin{pmatrix} \chi_1  \\ \chi_2 \end{pmatrix} \Big|  \begin{pmatrix}   \Delta \\ {\overline \Delta}\end{pmatrix}\right\rangle
 + \left\langle \begin{pmatrix} \chi_1  \\ \chi_2 \end{pmatrix} \Big| \begin{pmatrix}  r^\eps(t)  \\ -\overline  r^\eps(t)  \end{pmatrix} \right\rangle ,\;\;
 \Delta(0)=0,
 \ee
 with $ r^\eps(t) =O(\| \Delta(t)\|^2)$ 
 and 
 $$\displaylines{
 F(t)= F_0(t) +G(t) \ \mbox{with} \;
 F_0(t) = \begin{pmatrix}  H(t,[\omega(t)]) & 0 \\ 0 & -  H(t,[\omega(t)]) \end{pmatrix},\cr
 \;\; G(t)=\sum_{j=1}^p 
\omega_j(t) \begin{pmatrix}
   |v_j(t)\rangle \langle e_j| &   | v_j(t)\rangle \langle e_j|\\
 -  | v_j(t)\rangle \langle e_j |&-   | v_j(t)\rangle \langle e_j|
 \end{pmatrix}
 = \sum_{j=1}^p   \omega_j (t) \left| \begin{matrix} v_j(t) \\ - v_j (t)\end{matrix} \right\rangle 
 \left\langle \begin{matrix} e_j \\  e_j \end{matrix} \right|  .\cr}$$
 The conjugates do not appear in the definition of $F_0$, $F$ and $G$  since assumption {\bf R$_0$} entails the fact that $H(t,x)$ is real.

\medskip 

 {To analyse  the domain of $F(t)$, it is useful to see $F_0(t)$ as a perturbation of 
$F_0 = \begin{pmatrix}  H_0 & 0 \\ 0 & -  H_0 \end{pmatrix}$
by writing  $F_0(t)=F_0+B(t)$ with 
 $B(t) = \begin{pmatrix} W(t,[\omega(t)])   & 0 \\ 0 & -  W(t,[\omega(t)])  \end{pmatrix}$ bounded, self-adjoint and smooth in $t$.
 Indeed, this shows  shows that $F_0(t)$  is self-adjoint on $\tilde \cD:= \cD\times \cD$ and
has domain~$\tilde \cD$, and the same is true for $F(t)$ since $G(t)$ is also bounded.   We will also use this decomposition  to analyse the existence of a two-parameter semigroup associated with $F(t)$. }

\medskip 

  In the next three paragraphs, we develop the arguments of the proof paying attention to the difficulties induced by the fact that $H_0$, and thus~$F(t)$ are unbounded. As a fundamental preliminary, we first prove  the existence of an evolution semigroup propagator associated with the operator $F(t)$.   {Then,  the first step consists in  
  proving that   the (unbounded) spectrum of $F(t)$ consists in real eigenvalues that are all simple, except the eigenvalue zero, 
  and the second step in constructing the associated adiabatic approximate propagator as in Lemma~\ref{lem:adiab} and on its properties.}

\subsection{Existence of a two-parameter semigroup generated by  $F(t)$}\label{sec:proF(t)}

Using the latter remark, we get the following regularity result on the solutions to the linear part of the equation for $({\Delta(t)},  \overline{\Delta(t)})$ in $\cH\times\cH$.

\begin{lem}\label{domT} Let ${\mathcal T}$ be an interval such that $0\in\cT$ and let $F(t)=F_0+B(t)+G(t)$ such that $F_0$ is self-adjoint on $\tilde \cD=\cD\times\cD$ and $B(t)+G(t)$ defined for all $t\in{\mathcal T}$ is ${\mathcal C}^\infty$ and bounded.  Then, the equation 
$$
i\eps\partial_t T^\eps(t,s)= F(t) T^\eps(t,s),\;\; T^\eps(s,s)={\rm Id},
$$
admits a unique  strong solution with values in $\tilde\cD$, that is $C^1$ in time. Moreover, the same is true for the equation
\be\label{diffeqTstar}
i\eps\partial_t {T^\eps(t,s)}^*= - {T^\eps(t,s)}^*F^*(t),\;\; {T^\eps(s,s)}^*={\rm Id}.
\ee
 \end{lem}
 
 \begin{proof} The first statement follows from Thm X.70 in \cite{RS}, see also \cite{Kr}: the regularity assumption in time of $F(t)$ is satisfied thanks to {\bf R$_{1}$} so we need to show that for all fixed $t\in \cT_0$, $F(t)$ generates a contraction semigroup on $\cH\times \cH$. The operator $F_0(t)$ being self-adjoint on $\tilde\cD$, it generates a unitary group on $\cH\times \cH$. Since $G(t)$ is bounded, $F(t)=F_0(t)+G(t)$ generates a strongly continuous semigroup $S(s)_{s\geq t}$ (see Thm III.1.3 in \cite{EN}) which satisfies  $\|S(s)\|\leq {\rm e}^{\| G(t) \|s}$ in the operator norm of $\cH\times\cH$. By rescaling, $F(t)-\| G(t) \|\un$, defined on $\tilde\cD$, generates a contraction semigroup, so that Thm X.70 in \cite{RS} applies and the first statement follows.

\medskip 

 Since the existence of a strong derivative of $T(t,s)$ on $\tilde\cD$ does not imply directly the same for ${T(t,s)}^*$, we resort to the following decomposition: we write again $F(t)=F_0+A(t)$, where $A(t)=B(t)+G(t)$ and define the bounded operator $\Theta^\eps$ by
 $$
 \Theta^\eps(t,s)=e^{{i\over \eps}t F_0}T^\eps(t,s)e^{-{i\over \eps}s F_0}, \ \ \mbox{s.t.} \ \  \Theta^\eps(t,s)^{-1}= \Theta^\eps(s,t).
 $$
It satisfies the strong differential equation on $\cH\times\cH$
$$
i\eps\partial_t \Theta^\eps(t,s)= \tilde A^\eps(t) \Theta^\eps(t,s),\;\; \Theta^\eps(s,s)={\rm Id}, \ \ \mbox{with} \ \ \tilde A^\eps(t) = e^{{i\over \eps}t F_0}A(t)e^{-{i\over \eps}t F_0}.
$$
The generator $\tilde A^\eps(t)$ is strongly continuous on $\tilde\cD$ and  satisfies $\| \tilde A^\eps(t) \| = \|A(t)\|$ for all $t\in \bar \cT$. Hence we
can write $ \Theta^\eps(t,s)$ as a norm convergent Dyson series, uniformly in $t\in \bar \cT$, where the integrals are understood in the strong sense
$$
 \Theta^\eps(t,s)=\sum_{j\in\N} \Theta^\eps_j(t,s), \ \  \Theta^\eps_j(t,s)=\left(-{i\over \eps}\right)^j\int_s^t \int_s^{u_j}\dots \int_s^{u_2}\tilde A^\eps(u_j)\tilde A^\eps(u_{j-1})\dots \tilde A^\eps(u_1) du_{1}\dots du_{j-1}du_j. 
$$
The relation for $j\geq 1$,
$$
\Theta^\eps_j(t,s)=-{i\over \eps}\int_s^t \tilde A^\eps(u) \Theta^\eps_{j-1}(u,s)  du
$$
allows to prove by induction that $ t\mapsto \Theta^\eps_j(t,s)$ is continuous in norm and, for all $\ffi\in\cH\times\cH$
$$
i\eps\partial_t \Theta^\eps_j(t,s)\ffi = \tilde A^\eps(t) \Theta^\eps_{j-1}(t,s) \ffi.
$$
Hence $ t\mapsto {\Theta^\eps_j(t,s)}^*$ is norm continuous as well, and the same is true for $ {\Theta^\eps(t,s)}^*=\sum_{j\in\N} {\Theta^\eps_j(t,s)}^*$. 
Moreover, ${\Theta^\eps_j(t,s)}^*\psi$, for any $\psi\in\cH\times\cH$,  satisfies for any $\ffi\in\cH\times\cH$
\begin{align}\label{strstar}\nonumber
\bra \ffi | {\Theta^\eps_j(t,s)}^*\psi\ket&=\left\langle -{i\over \eps}\int_s^t \tilde A^\eps(u) \Theta^\eps_{j-1}(u,s) \ffi du \Big| \psi \right\rangle={i\over \eps} \int_s^t \bra 
 \tilde A^\eps(u) \Theta^\eps_{j-1}(u,s) \ffi | \psi \ket du\\
 &={i\over \eps} \int_s^t \bra \ffi |  { \Theta^\eps_{j-1}(u,s)}^* { \tilde A^\eps(u)}^*\psi \ket du.
\end{align}
Since $ { \tilde A^\eps(t)}^*= e^{it F_0/\eps}A^*(t)e^{-it F_0/\eps}$, where $t\mapsto A(t)$ is norm continuous, we get that   $t\mapsto  { \tilde A^\eps(t)}^*$ is strongly continuous, see e.g. \cite{Kr},  and so is $t\mapsto  { \Theta^\eps_{j-1}(t,s)}^* { \tilde A^\eps(t)}^*$. Hence we deduce 
from (\ref{strstar}) that for any $\psi\in \cH\times\cH$,
$$
 {\Theta^\eps_j(t,s)}^*\psi={i\over \eps }\int_s^t  { \Theta^\eps_{j-1}(u,s)}^* { \tilde A^\eps(u)}^*\psi du,
$$
 which, as above, implies for all $j\geq 1$ and all $\psi\in \cH\times\cH$,
$$
 i\eps\partial_t {\Theta^\eps_j(t,s)}^*\psi=-  { \Theta^\eps_{j-1}(t,s)}^* { \tilde A^\eps(t)}^*\psi .
$$
 This differential identity allows then to get the key property
$$
  i\eps\partial_t {\Theta^\eps (t,s)}^*\psi=- { \Theta^\eps(t,s)}^* { \tilde A^\eps(t)}^*\psi,
$$
 which derives from the Dyson representation for ${\Theta^\eps (t,s)}^*$. Therefore, $ {T^\eps(t,s)}^*=e^{-is F_0/\eps}{ \Theta^\eps(t,s)}^*e^{it F_0/\eps}$
 is strongly continuously differentiable in $t$ on $\tilde\cD$, since all operators in the composition are, and (\ref{diffeqTstar}) holds. 
 \qed
 \end{proof}

\subsection{The spectrum of $F(t)$}\label{sec:specF(t)}

 {We prove here that 
  the spectrum of $F(t)$ has the required properties
  for $\delta$ small enough.}  
  
  \medskip 
  
In that purpose,   we use that, as a consequence of the hypothesis~{\bf S}$_2$: 
$$\forall (t,x)\in \cT_0\times \cX^p, \;\;\sigma(H(t,x))\cap \sigma(-H(t,x))=\{0\}.$$
Note that  the operator $F_0(t)$ satisfies the assumptions of Theorem~4.15a in~\cite{K2}, with the generalization stated in b) of Remark~4.16a. We deduce that the spectrum of~$F(t)$ consists in a sequence of eigenvalues
\[
\dots<-\ell_{j}(t) <\cdots <-\ell_2(t)<-\ell_1(t)<0<\ell_1(t)<\ell_2(t)<\cdots< \ell_{j}(t)<\dots,
\]
where $\pm\ell_j(t)$ are simple eigenvalues, while $\ell_0(t)\equiv 0$ has multiplicity 2, with zero eigennilpotent. Each 
$\ell_j(t)$ corresponds to a unique eigenvalue of the unperturbed operator $F_0(t)$ determined by $H(t,[\omega(t)])$. We denote 
those corresponding eigenvalues of $F_0(t)$ by $\pm \lambda_j(t)$, {$j\in \N$} (recall that the labelling of the $\lambda_j$s may differ from that
of the eigenvalues of $H$.
Besides, there exists a constant $c$  such that $\forall t\in{\mathcal T}_0, \forall j\in\Z,\;\; $
\begin{align}\label{controlel}
|\ell _{j+1}(t)-\ell_j (t)|\geq c|j|^\alpha, \ \
\ell_j(t)=\ell_j(0)+(\ell_j(t)-\ell_j(0)),  \ \mbox{and }  \ \sup_{j\in\Z}\sup_{t\in \overline{\cT_0}}|\ell_j(t)-\ell_j(0)|<\infty.
\end{align}
Moreover, 
\begin{equation}\label{toto4}
\forall t\in{\mathcal T}_0, \;\; \forall j\in \Z,\;\; |\dot \ell_j(t)|\leq c,
 \end{equation}
which derives from the observation 
$F(t) {\mathbb P}_j(t)=\ell_j(t) {\mathbb P}_j(t)$:  By differentiation, 
$$\dot A(t) {\mathbb P}_j(t) + F(t) \dot {\mathbb P}_j(t)=\dot \ell_j(t) {\mathbb P}_j(t)+\ell_j(t) \dot {\mathbb P}_j(t)$$
whence, 
using ${\mathbb P}_j(t)\dot {\mathbb P}_j(t){\mathbb P}_j(t)=0$, one gets for the rank one projector ${\mathbb P}_j(t)$, $j\neq 0$,
$$\dot \ell_j(t){\mathbb P}_j(t)= {\mathbb P}_j(t) \dot B(t) {\mathbb P}_j(t).$$

\medskip 

The fact that $F(t)$ is a slightly non-selfadjoint operator in the sense of Section V.5 in \cite{K2} allows us to apply Theorem 4.16 in \cite{K2} and Remark 4.17 following it, to get the following spectral decomposition, under our assumption $\alpha>1/2$ in {\bf S}$_1$, and for $\delta_0$ small enough:
\begin{align}\label{decspecub}
F(t)&=\sum_{j=-\infty}^{\infty} \ell_j(t) {\mathbb P}_j(t), \ \mbox{with the convention} \ \ell_{-|j|}(t)=-\ell_{|j|}(t), \ \mbox{where}\\
\label{def:PhijPsij}
{\mathbb P}_j(t)&=|\Psi_j(t)\ket\bra \Phi_j(t)|, \ j\neq 0,  \ \ \ {\mathbb P}_0(t)=\sum_{\sigma=1,2}|\Psi_0^\sigma(t)\ket\bra \Phi_0^\sigma(t)|,
\end{align}
with $\{\Psi_j(t), \Phi_j(t)\}_{j\neq 0}\cup \{\Psi_0^\sigma(t), \Phi_0^\sigma(t)\}_{\sigma\in\{1,2\}}$ a biorthogonal family of vectors, with 
$\|\Psi_j\|=\|\Psi_0^\sigma\|=1$. The sum (\ref{decspecub}) is understood in the strong convergence sense on the time independent domain 
\be\label{dtilde}
\tilde \cD =\{\chi = \sum_{j\in \Z} \alpha_j \Psi_j(0) \ \mbox{s.t.} \ \sum_{j\in \Z} |\alpha_j \ell_j(0)|^2<\infty\} \subset \cH\times \cH.
\ee
Indeed, Theorem 4.16 in \cite{K2} states that the normalised basis $\{\Psi_j(t) \}_{j\in\Z}$ is a Riesz basis, and 
Theorem~3.4.5 in \cite{D}, giving a characterisation of Riesz basis, allows for the explicit description of the domain~$\tilde \cD$. 
In particular, there exist $0<C, M<\infty $ such that for all~$t\in\cT$,
\begin{align}
\label{toto3}
&\big\|\sum_{j\in I} {\mathbb P}_j(t)\big\|\leq M, \ \ \forall I\in \Z,\\
\label{normriesz}
& C^{-1}\|\chi\|^2\leq \sum_{j\in\Z} \|{\mathbb P}_j(t)\chi\|^2 \leq C\|\chi\|^2   ,\ \ \forall \chi\in \cH\times \cH ,
\end{align}
where
$$\displaylines{
\chi=\sum_{j\in\Z \atop j\neq 0}\alpha_j(t)\Psi_j(t) + \sum_{\sigma=1}^2\Psi_0^\sigma(t)\alpha_0^\sigma(t),\cr
 {\mathbb P}_j(t)\chi=\bra \Phi_j(t) | \chi\ket \Psi_j(t),\;\;\forall  j\neq 0,\cr
{\mathbb P}_0(t)\chi=\sum_{\sigma=1}^2\bra \Phi_0^\sigma(t) | \chi\ket\Psi_0^\sigma(t).\cr}$$
Note that the domain
of $H_0$ 
is 
$$\cD=\{\ffi=\sum_{k\in \N} \beta_k \ffi_k\ \mbox{s.t.} \ \sum_{k\in \N} |\beta_j \lambda_k|^2<\infty\},$$ 
where $(\lambda_k, \ffi_k)$ are the eigenvectors and eigenvalues of $H_0$.
The reader can refer to the paper~\cite{GZ}, for example, in which Riesz spectral systems are studied.

\subsection{The  adiabatic propagator and its properties} \label{sec:adiaF(t)}

 {We now focus on the construction of the adiabatic propagator as in Lemma~\ref{lem:adiab}. Since its proof follows that of the bounded case, 
we only have to focus on domain issues.}

\medskip 
 
 {In view of what we have done in the previous sections}, we can define, as in the bounded case,
the dynamical phase operator~$\Phi^\eps(t,s)$ (see~(\ref{dynpha}) and~(\ref{paraltrans}))
\be\label{dynpha2}
\Phi^\eps(t,s)= \sum_{j=-\infty}^{\infty} {\mathbb P}_j(0) {\rm e}^{-{i\over \eps}\int_s^t \ell_j(\sigma) d\sigma}, \ \mbox{s.t.} \ \Phi^\eps(t,s)^{-1}=\Phi^\eps(s,t),
\ee
which is a family of uniformly bounded operators that map $\tilde \cD $ on $\tilde \cD $, thanks to  (\ref{normriesz}). 
At this point, further making use of (\ref{controlel}) and of the fact that $|(e^{ix}-1)/x|$ is uniformly bounded in $x\in \R$, one sees by a dominated convergence argument that $t\mapsto \Phi^\eps(t,s)$ is also a strongly continuously differentiable two-parameter evolution operator on $\tilde \cD $, where (\ref{diffeqphi}) holds. 

\medskip 

We also define the intertwining operator $W(t)$ given by
\be\label{paraltrans2}
i\partial_t W(t) = K(t) W(t) ,\;\;W(0)={\rm Id},\;\;
{\rm with}\;\; 
K(t)=i \sum_{j=-\infty}^{\infty} \dot {\mathbb P}_j(t){\mathbb P}_j(t) .
\ee
It is shown in Proposition~3.1 and Lemma~3.2 of \cite{J2}, 
that as soon as $\alpha >0$, $K(t)$ is well defined, ${\mathcal C}^\infty$, and $W(t)$ satisfies the intertwining property (\ref{commutation}) with each of the projectors.

\medskip 

Actually, theses properties of~$W$ are shown in~\cite{J2}  for orthogonal projectors~${\mathbb P}_j(t)$. However, as a routine inspection reveals, the proofs hold {\it mutatis mutandis} in the non selfadjoint case, provided the growing gap assumption {\bf S} holds, and the 
resolvent $(F(t)-z)^{-1}$ can be bounded in an approximate way by the inverse of the distance to the spectrum. Our perturbative framework, characterised by $\delta_0$
 small ensures that this is the case. 

 \medskip

We then introduce the bounded family of operators 
\begin{equation}\label{def:Veps2}
V^\eps(t,s)=W(t) \Phi^\eps(t,s) W(s)^{-1},
\end{equation}
which map $\tilde \cD $ on $\tilde \cD $ and satisfy $V^\eps(t,s)^{-1}=V^\eps(s,t)$, together with
\begin{align}\nonumber
&V^\eps(t,s){\mathbb P}_j(s)={\mathbb P}_j(t)V^\eps(t,s)= W(t) {\mathbb P}_j(0) {\rm e}^{-{i\over \eps}\int_s^t \ell_j(\sigma) d\sigma} W(s)^{-1}.
\end{align}
The latter intertwining property implies that $W(t)$ maps $\tilde \cD $ on $\tilde \cD $: 
From (\ref{decspecub})  and the definition of~$\Phi_j(t)$ and~$\Psi_j(t)$ in~(\ref{def:PhijPsij}),  for $j\neq 0$, 
$$
W(t)|\Psi_j(0)\ket\bra \Phi_j(0)|=|\Psi_j(t)\ket\bra \Phi_j(t)| W(t) \ \Rightarrow \ W(t)\Psi_j(0)=\Psi_j(t) \bra \Phi_j(t)| W(t)\Psi_j(0)\ket,
$$
so that we have the following property: if $\chi\in \tilde D$, see (\ref{dtilde}), with coefficient $\alpha_j=\bra \Phi_j(0)| \chi\ket$, $j\neq 0$, in the basis at time $0$, then $W(t)\chi$ has an expansion in the basis at time $t$ with coefficients $\alpha_j(t)=\bra \Phi_j(t)|W(t)\Psi_j(0)\ket \alpha_j$, $j\neq 0$, where 
$|\bra \Phi_j(t)|W(t)\Psi_j(0)\ket|$ is uniformly bounded in $j\neq 0$, thanks to (\ref{toto3}).

\medskip

We now describe the adjustments requested to argue as in Section~\ref{sec:adiab} to prove the analogue of Lemma \ref{lem:adiab}, that is 
$$T^\eps(t,s)= V^\eps(t,s) +O_{t,s} (\eps).$$
We recall that the differential equation (\ref{semiF}) has to be understood in the strong sense on $\tilde \cD$, and~$T^\eps(t,s)$ is ${\mathcal C}^1$ on~$\tilde \cD$ and maps  $\tilde \cD$ on  $\tilde \cD$, according to Lemma \ref{domT}. 
Analogously, $V^\eps(t,s)$ satisfies~(\ref{defadiabV}) in the strong sense on $\tilde D$, and the same holds for $\Omega^\eps(t,s)$ defined by (\ref{diffomeg}). Then, integration by parts on the integrand of (\ref{diffomeg}) is to be understood in the strong sense, on vectors of $\tilde D$.
To deal with (\ref{shiftint}),  one notes that 
(\ref{redeqdiff}) holds  in the strong sense on $\tilde \cD$, with $\tilde F_j(t)=\tilde F(t)-\ell_j(t)\un$ the closed operator on $\tilde D$ obtained by extending the summation to $k\in\Z$ in (\ref{defftilde}). Similarly, its reduced resolvent on $(\un-{\mathbb P}_j(0))\cH\times\cH$ simply  reads
$\tilde R_j(t)=\sum_{k\in\Z\atop k\neq j} {\mathbb P}_k(0)/(\ell_k(t)-\ell_j(t)).$ Note that thanks to (\ref{normriesz}) and the spectral behaviours (\ref{controlel}) and (\ref{toto4}), we have with the notation $\bra j\ket=(1+j^2)^{1/2}$
$$
\max \{\| \tilde R_j(t) \|, \| \partial_t \tilde R_j(t) \| \}  \leq c\bra j\ket^{-\alpha}, 
$$
for some constant $c$ uniform in $t\in \overline{\cT}$, that may change from line to line below.
Using this estimate in the integration by parts formula (\ref{ippformul}) we now get 
\begin{align}
\|  {\mathbb P}_j(s)\Omega(t,s)- {\mathbb P}_j(s)\| \leq c \,{\eps}\,\bra j\ket^{-\alpha}||| \Omega |||,
\end{align}
where $||| \Omega |||=\sup_{(s,t)\in \overline \cT}\|\Omega(t,s)\|$. Therefore, since $2\alpha >1$,
$$
\sup_{(s,t)\in \overline \cT}\| \Omega(t,s)-\un \|^2\leq c \,\eps^2 ||| \Omega |||^2,
$$
from which we get, as in Section~\ref{sec:adiab}, that for $\eps < \eps_0$, $\eps_0$ independent of $t$, 
$$
 ||| \Omega -\un |||=O(\eps).
$$
In turn, this proves Lemma \ref{lem:adiab} in our current unbounded context.

\medskip 

Given the observations above, we also note that the arguments used in proof of  {Corollary}~\ref{lem:dotomega} are valid in the unbounded case as well.

\subsection{Conclusion of the proof of Theorem~\ref{thm:unbounded}\ (2)} \label{sec:conclusion}

We set 
$$\delta^\eps_\tau= \sup_{t\in [0,\tau]} \, \sup_{\chi\in\cD,\,\|\chi\|=1}| \bra \chi|\Delta(t)\ket| = \sup_{t\in [0,\tau]} \,  \| \Delta(t) \|.$$
In particular, since $e_j\in\cD$ for all $j\in\{1,\cdots,p\}$, we have 
$$\forall t\in[0,\tau],\;\; |\Delta_j(t)| \leq \delta^\eps_\tau.$$  
Besides, for any  family of bounded operators  $C(t)$ on $\cH\times\cH$, for $\begin{pmatrix} \chi_1 \\ \chi_2 \end{pmatrix}\in  \tilde \cD $ normalized  and for $0\leq s\leq t\leq \tau$, using (\ref{lem:adiab}) 
$$\left|\left\langle T^\eps(t,s)^* \begin{pmatrix} \chi_1 \\ \chi_2 \end{pmatrix}\, \Big| C(s) \begin{pmatrix} \Delta(t)\\ \bar \Delta(t) \end{pmatrix} \right\rangle \right| \leq \sup_{s\in[0,\tau] } \| T^\eps(t,s) C(s)\|  \sqrt 2 \delta_\tau^\eps \equiv \theta \sup_{s\in[0,\tau] } \| C(s)\| \delta_\tau^\eps.$$
We then deduce from~(\ref{prop:reps}) that there exists a constant $b>0$ such that for any $\begin{pmatrix} \chi_1 \\ \chi_2 \end{pmatrix}\in  \tilde \cD $ and $0\leq s\leq t\leq \tau$,
\begin{equation}\label{estimateb}
| \bra \chi_1| r^\eps(s)\ket |+ | \bra \chi_2  |  \overline r^\eps(s) \ket |\leq  b ( \delta^\eps_\tau)^2.
\end{equation}
We observe that $T^\eps(s,t)^*= (T^\eps(t,s)^{-1})^*$ satisfies in the strong sense on $\tilde \cD $
$$i\eps\partial_t (T^\eps(s,t)^*) = F(t)^* T^\eps(s,t)^*,\;\; T^\eps(s,s)^*=\un.$$
In view of (\ref{weakdeldot}),  
for any $\chi_1,\chi_2\in\cD$, 
$$\displaylines{
i\eps \partial_t \left\langle T^\eps(s,t)^* 
\begin{pmatrix} \chi_1 \\ \chi_2 \end{pmatrix} \Big|
\begin{pmatrix} \Delta(t)\\ \bar \Delta(t) \end{pmatrix}
\right\rangle
= - \left\langle  F(t) ^*T^\eps(s,t)^* 
\begin{pmatrix} \chi_1 \\ \chi_2  \end{pmatrix} \Big| 
\begin{pmatrix} \Delta(t)\\ \bar \Delta(t) \end{pmatrix}
\right\rangle
\hfill\cr\hfill
+ \left\langle  F(t)^* T^\eps(s,t)^* 
\begin{pmatrix} \chi_1 \\ \chi_2  \end{pmatrix} \Big| 
\begin{pmatrix} \Delta(t)\\ \bar \Delta(t) \end{pmatrix}
\right\rangle
+\left\langle T^\eps(s,t)^* 
\begin{pmatrix} \chi_1 \\ \chi_2  \end{pmatrix} \Big| 
\begin{pmatrix} r^\eps(t)\\ -\overline r^\eps(t) \end{pmatrix} \right\rangle
-i\eps  \left\langle T^\eps(s,t)^* \begin{pmatrix} \chi_1 \\ \chi_2 \end{pmatrix} \Big|
\begin{pmatrix} \dot \omega(t)\\ {\dot  \omega}(t) 
\end{pmatrix}\right\rangle,\cr}$$
where the first term of the right hand side comes from the equation of $T^\eps(s,t)^*$, 
and the second term comes from the fact that $\Delta$ satisfies the equation in the weak sense, 
making use of $T^\eps(s,t)^*:\widetilde \cD\ra \widetilde \cD$.
Therefore, integrating between $0$ and $s$, we obtain 
$$
\left\langle 
\begin{pmatrix}\chi_1 \\ \chi_2  \end{pmatrix} \Big|
\begin{pmatrix} \Delta(s)\\ \bar \Delta(s) \end{pmatrix}
\right\rangle
= {1\over i\eps} \int_0^s 
 \left\langle T^\eps(s,t)^* 
\begin{pmatrix}\chi_1 \\ \chi_2  \end{pmatrix} \Big| \begin{pmatrix} r^\eps(t)\\ -\overline r^\eps(t) \end{pmatrix}\right\rangle dt
- \int_0^s   \left\langle T^\eps(s,t)^* \begin{pmatrix}\chi_1 \\ \chi_2  \end{pmatrix} \Big|
\begin{pmatrix} \dot \omega(t)\\ {\dot  \omega}(t) 
\end{pmatrix}\right\rangle
 dt.
$$
Since $ T^\eps(s,t)^* \begin{pmatrix}\chi_1 \\ \chi_2  \end{pmatrix}\in\cD\times\cD$, we can use estimate~(\ref{estimateb}) for normalised $\begin{pmatrix}\chi_1 \\ \chi_2  \end{pmatrix} $ and there exists $\widetilde b$ such that  
$$\left| \int_0^s 
 \left\langle T^\eps(s,t)^* 
\begin{pmatrix}\chi_1 \\ \chi_2  \end{pmatrix} \Big| 
\begin{pmatrix} r^\eps(t)\\ -\overline r^\eps(t) \end{pmatrix}
\right\rangle dt\right| \leq {\tilde b} |s|  (\delta^s_\eps)^2.$$
Besides, 
\begin{eqnarray*}
\int_0^s   \left\langle T^\eps(s,t)^* \begin{pmatrix}\chi_1 \\ \chi_2  \end{pmatrix} \Big|
\begin{pmatrix} \dot \omega(t)\\ {\dot  \omega}(t) 
\end{pmatrix}\right\rangle
 dt&=&
 \int_0^s   \left\langle \begin{pmatrix}\chi_1 \\ \chi_2  \end{pmatrix} \Big| T^\eps(s,t)
\begin{pmatrix} \dot \omega(t)\\ {\dot  \omega}(t) 
\end{pmatrix}\right\rangle dt\\
& =&  \int_0^s   \left\langle \begin{pmatrix}\chi_1 \\ \chi_2  \end{pmatrix} \Big| V^\eps(s,t)
\begin{pmatrix} \dot \omega(t)\\ {\dot  \omega}(t) 
\end{pmatrix}\right\rangle dt +O_s(\eps) =O_s(\eps)
\end{eqnarray*}
by Lemma~\ref{lem:adiab} and~\ref{lem:dotomega}.
Finally, 
 by 
choosing $\chi_1=\chi$, $\chi_2=0$, we obtain  that there exists constants $a,\widetilde b>0$, uniform in $0\leq s \leq \tau$
$$|\langle \chi, \Delta(s)\rangle| \leq  a\eps + {\tilde b\over \eps} |s|  (\delta^s_\eps)^2, $$
whence
$$\delta^\tau_\eps \leq a\eps + {b\over \eps} \tau (\delta^\tau_\eps)^2,$$
which allows to conclude the proof. \qed


\appendix
\section{Appendix A}
 
According to Remark~\ref{rem:appendixA}, we provide here an argument showing the spectrum of $F(t)$ is not necessarily real if $H(t,x)$ is real.
We consider a smooth Hamiltonian $\mathbb R\times \mathbb R^2 \ni (t,x)\mapsto H(t,x)$ on a Hilbert space~$\cH$, and of the form
\be\label{ex2}
H(t,x)=\lambda_1(t,x)P_1(t,x)+\lambda_2(t,x)P_2(t,x) \ \ \mbox{with} \ \ \sum_{j=1}^3 P_j(t,x)\equiv \un,
\ee
with the assumption that the eigenvalue $0$ is simple and that the $\lambda_j(t,x)$ are of arbitrary multiplicities ($j=1,2$). With the assumptions of Section~\ref{sec:32},  that means $N'=p=2$ and, dropping the arguments 
$(t,[\omega(t)])$ in the variables, the Aronszajn-Weinstein determinant (\ref{AWdet}) takes the form
$$
w(z)=\det \left(\delta_{jk}+\sum_{l=1}^2 \frac{2\lambda_l\omega_j\bra e_j | P_l v_k\ket}{(\lambda_l-z)(\lambda_l+z)}\right)_{1\leq j,k\leq 2}.
$$
Introducing $q_{12}(z)=(\lambda_1-z)(\lambda_1+z)(\lambda_2-z)(\lambda_2+z)$ and $q_j(z)=(\lambda_j-z)(\lambda_j+z)$, $j=1,2$, we have
$$
w(z)=\frac{1}{q_{12}(z)}\det \left(q_{12}(z)\un + 2\begin{pmatrix}  \omega_1 \bra e_1 | (\sum_{l=1}^2 P_l \lambda_l q_{\bar l}(z)) v_1\ket  & \omega_1 \bra e_1 | (\sum_{l=1}^2 P_l \lambda_l q_{\bar l}(z)) v_2\ket \cr \omega_2 \bra e_2 | (\sum_{l=1}^2 P_l \lambda_l q_{\bar l}(z)) v_1\ket & \omega_2 \bra e_2 | (\sum_{l=1}^2 P_l \lambda_l q_{\bar l}(z)) v_2\ket \end{pmatrix} \right),
$$
where $\bar 1 = 2$ and $\bar 2 =1$. By assumption, all matrix elements are real-valued. If $z_0\in \mathbb R\setminus \{\lambda_1, \lambda_2\}$ is a zero of $w(z)$, that is a real eigenvalue of $F$, that means $-q_{12}(z_0)/2$ is a real nonzero eigenvalue of the matrix
$$
b(z_0)=
\begin{pmatrix}  \omega_1 \bra e_1 | (\sum_{l=1}^2 P_l \lambda_l q_{\bar l}(z_0)) v_1\ket  & \omega_1 \bra e_1 | (\sum_{l=1}^2 P_l \lambda_l q_{\bar l}(z_0)) v_2\ket \cr \omega_2 \bra e_2 | (\sum_{l=1}^2 P_l \lambda_l q_{\bar l}(z_0)) v_1\ket & \omega_2 \bra e_2 | (\sum_{l=1}^2 P_l \lambda_l q_{\bar l}(z_0)) v_2\ket \end{pmatrix} \in M_2(\mathbb R).
$$
This requires $(\tr \ b(z_0))^2-4\det b(z_0)>0$, which is not granted for a generic matrix in $M_2(\R)$. While $b(z_0)$ is not completely arbitrary, it doesn't necessarily possess the symmetries that enforce this, as we argue below. Hence, the existence of nonzero real eigenvalues for $F$ cannot be inferred from the sole requirement that
$H$ is real. 

\medskip

To be more quantitative, assume the eigenvalue $\lambda_2$ of $H(t,x)$ is independent of $(t,x)$. Thus $\omega(t)$ is independent of $\lambda_2$ that we will consider as a large parameter. Consider $t$ fixed and $z_0$ in the vicinity of $\lambda_1(t,[\omega(t)])$, assumed to be of order one. Then, for $\lambda_2>0$ large, we have $q_{12}(z_0)=\lambda_2^2(\lambda_1^2-z_0^2)+O(1)$, $q_2(z_0)=\lambda_2^2+O(1)$, $q_1(z_0)=O(1)$ so that
$
\sum_{l=1}^2 P_l \lambda_l q_{\bar l}(z_0)=P_1\lambda_1\lambda_2^2+O(\lambda_2)
$  
and 
\be\label{leadbz}
b(z_0)=\lambda_1\lambda_2^2 \begin{pmatrix}  \omega_1 \bra e_1 | P_1 v_1\ket  & \omega_1 \bra e_1 |P_1 v_2\ket \cr \omega_2 \bra e_2 | P_1  v_1\ket & \omega_2 \bra e_2 |P_1  v_2\ket \end{pmatrix}+ O(\lambda_2).
\ee
The condition $(\tr \ b(z_0))^2-4\det b(z_0)>0$  for $\lambda_2$ large, is equivalent to saying the $z_0$ independent leading order matrix in (\ref{leadbz}) has real eigenvalues, 
{\it i.e.} to having
\be\label{cotrdet}
(\omega_1 \bra e_1 | P_1 v_1\ket - \omega_2 \bra e_2 | P_1 v_2\ket )^2+4 \omega_1\omega_2 \bra e_1 | P_1 v_2\ket \bra e_2 | P_1 v_1\ket )>0.
\ee
Recall that given $(\omega_1, \omega_2)=[\omega]$,  the operators $H([\omega])$, $P_1([\omega])$ and $\partial_{x_j}H([\omega])$, $j=1,2$ are fixed,
as is $\omega=\ffi([\omega])$. 
Hence, the same is true for
\be\label{uj}
u_j=P_1([\omega])v_j=P_1([\omega])\partial_{x_j}H([\omega])\omega\equiv K_j([\omega])\omega, \ \ \mbox{with} \ \ \bra u_j | \omega \ket=0, \ \ j=1,2,
\ee  
so that (\ref{cotrdet}) reads
\be\label{realev}
(\omega_1 \bra e_1 | u_1\ket - \omega_2 \bra e_2 | u_2\ket )^2+4 \omega_1\omega_2 \bra e_1 | u_2\ket \bra e_2 | u_1\ket )>0.
\ee
For generic vectors $\{e_1, e_2, \omega, u_1, u_2\}$ satisfying (\ref{uj}), the above condition needs not be true.
Actually, for any real unitary operator $R$ such that $R\omega=\omega$, we have $\omega_j=\bra \omega | e_j\ket=\bra \omega | R e_j\ket$, so that
$\{f_1,f_2\}=\{R e_1, R e_2\}$ forms another orthonormal family defining the nonlinearity of the problem, keeping $\omega_j$, $j=1,2$ fixed. It can be shown that if
(\ref{realev}) holds for $\{e_1, e_2, \omega, u_1, u_2\}$, with $\dim({\mathbb C \omega })^\perp \geq 3$, $\omega_j\neq 0$, and $0<|\bra u_1|u_2\ket|<\|u_1\|\|u_2\|$, a real unitary $R$ leaving $\omega$ invariant can be chosen to that (\ref{realev}) is false for $\{f_1, f_2, \omega, u_1, u_2\}$. The idea consists in discussing the restriction of $R$ to $({\mathbb C \omega })^\perp$ so that the orthonormal vectors $\{f_1,f_2\}$ have scalar products with 
$\{u_1,u_2\}$ which make (\ref{realev}) false.

 \section{Appendix B} 
 
Let us look for more general solutions to (\ref{exple}) and prove Lemma~\ref{lem:energycontent}. Reparametrising the time variable  
$t\mapsto s(t)=\int_0^t\gamma(u)du$ and writing $w(s(t))=v(t)$ allows us to get rid of the factor $\gamma(t)$,
$$
i\eps\partial_s \begin{pmatrix} w_1 \cr w_2 \end{pmatrix} = |w_1|^2\begin{pmatrix} w_2 \cr w_1 \end{pmatrix}.  
$$
Writing out $w_1(s)=x(s)+iy(s), w_2(s)=z(s)+it(s)$, we get the equivalent system 
$$
\left\{ \begin{matrix}
\eps \dot x=(x^2+y^2)t \cr
\eps \dot y=-(x^2+y^2)z \cr
\eps \dot z=(x^2+y^2)y \cr
\eps \dot t=-(x^2+y^2)x
\end{matrix}\right.
$$
It is readily checked that the three following expressions are constants of the motion
$$
x^2+t^2, \ \ y^2+z^2, \ \ xz+yt,
$$
so that the system can be solved by quadratures. Refraining from spelling out the solution in full generality, we consider solutions corresponding to the initial conditions
$$
y(0)=t(0)=0, \ \ x(0)>0, z(0)\neq 0.
$$
We get for all $s\in\R$ with $\alpha_\eps(s)=-x(0)z(0)s/\eps$
$$
\left\{ \begin{matrix}
x(s)=\frac{x(0)\cos(\alpha_\eps(s))}{\left(\cos^2(\alpha_\eps(s))+\big(\frac{x(0)}{z(0)}\big)^2\sin^2(\alpha_\eps(s))\right)^{1/2}} \cr
y(s)=\frac{x(0)\sin(\alpha_\eps(s))}{\left(\cos^2(\alpha_\eps(s))+\big(\frac{x(0)}{z(0)}\big)^2\sin^2(\alpha_\eps(s))\right)^{1/2}} 
\end{matrix}\right. , \\ 
\left\{ \begin{matrix}
z(s)=\frac{z(0)\cos(\alpha_\eps(s))}{\left(\cos^2(\alpha_\eps(s))+\big(\frac{x(0)}{z(0)}\big)^2\sin^2(\alpha_\eps(s))\right)^{1/2}} \cr
t(s)=\frac{x^2(0)\sin(\alpha_\eps(s))/z(0)}{\left(\cos^2(\alpha_\eps(s))+\big(\frac{x(0)}{z(0)}\big)^2\sin^2(\alpha_\eps(s))\right)^{1/2}}. 
\end{matrix}\right.
$$
In case $x(0)=1=\pm z(0)$, we recover (\ref{appex}), modulo the reparametrization of the time variable. In all other cases, noting that 
$\Re (w_1\overline{w_2})$ is conserved, we compute in the $s$ variable
$$
E_{w}(s)=2\frac{x(0)^3z(0)}{\cos^2(\alpha_\eps(s))+\big(\frac{x(0)}{z(0)}\big)^2\sin^2(\alpha_\eps(s))},
$$
which gives the result of the Lemma~\ref{lem:energycontent} with $\aleph(t)= \eps \alpha_\eps(s)$, $x(0)= v_1(0)$ and $z(0)= v_2(0)$.


 {

 \end{document}